\documentclass[aps,prl,twocolumn,10pt,superscriptaddress,nofootinbib,longbibliography,nobalancelastpage]{revtex4-2}
\usepackage{amsmath,amssymb,amsthm,mathtools,bm}
\usepackage{microtype}
\usepackage[hidelinks]{hyperref}
\newtheorem{theorem}{Theorem}
\newcommand{\Q}{\mathcal Q}
\newcommand{\id}{\mathbf 1}
\newcommand{\Tr}{\operatorname{Tr}}
\begin{document}
\title{No Finite NPA Level Characterizes the Complete Quantum Set in the Simplest Bell Scenario}
\author{Anubhav Chaturvedi}
\email{anubhav.chaturvedi@ug.edu.pl}
\affiliation{Institute of Theoretical Physics and Astrophysics, Faculty of Mathematics, Physics and Informatics, University of Gda\'nsk, 80-308 Gda\'nsk, Poland}
\date{September 16, 2026}
\begin{abstract}
The Navascués–Pironio–Acín (NPA) hierarchy provides increasingly tighter semidefinite outer approximations to the set of quantum correlations, while its dual programs give sum-of-squares certificates that can establish exact quantum bounds for Bell inequalities and yield self-testing relations. We show that no finite level of the NPA hierarchy can exactly characterize the set of quantum correlations, even in the simplest Bell scenario. For the symmetric doubly tilted CHSH inequality, the maximal quantum violation vanishes cubically as the tilting parameter approaches its critical value, while the error of every fixed NPA level relative to this violation diverges. We construct a family of two-qubit strategies for which the gap between the exact quantum bound and the attained Bell value, divided by the quantum violation, tends to the Motzkin polynomial in rescaled measurement angles. If the relative error remained bounded, the fixed-level NPA certificates would yield a sum-of-squares representation of this limiting polynomial plus a nonnegative constant, which is impossible.

\end{abstract}
\maketitle

Which correlations between distant measurements are compatible with quantum theory? The Navascu\'es--Pironio--Ac\'in (NPA) hierarchy addresses this question through a sequence of semidefinite programs~\cite{NPA2007,NPA2008,Tavakoli2024}. For a Bell inequality, successive levels give tighter upper bounds on its maximal quantum value, with dual sum-of-squares decompositions that can certify the exact bound and furnish the operator relations underlying self-testing proofs~\cite{BampsPironio2015,SupicBowles2020,KaniewskiCEQIP2018}. The hierarchy converges asymptotically, but this does not determine whether some finite level already recovers the entire quantum set.

The simplest Bell scenario has two parties, each choosing between two binary measurements, yet much of its quantum boundary is already captured remarkably early in the NPA hierarchy. Level $1$ gives the quantum CHSH bound $2\sqrt2$~\cite{CHSH1969,Tsirelson1980,NPA2007}, while level $1+AB$ gives the exact quantum bound of every one-sided tilted CHSH inequality~\cite{AMP2012,BampsPironio2015}. Even the entire four-correlator projection admits an exact finite semidefinite description~\cite{Landau1988,LeEtAl2023}. Retaining the marginal averages gives the full quantum set, whose boundary has been studied through analytical extremality criteria~\cite{MikosKaniewski2023,Ishizaka2025} and a characterization of its nonlocal extreme points~\cite{BarizienBancal2025}. None of these results determines whether a finite NPA level recovers the full quantum set.

Doubly tilted CHSH inequalities arise naturally from CHSH tests with inefficient detectors on both sides, after assigning no-click events to one of the existing outcomes~\cite{MiklinEtAl2022,GigenaEtAl2025}. In earlier work~\cite{GigenaEtAl2025}, we derived their exact quantum bounds and self-tested the optimal strategies using Jordan's lemma and Gr\"obner-basis elimination. These exact bounds identify a critical tilt at which the quantum and local bounds coincide. Our earlier numerical calculations found a level-$10$ NPA bound above the exact quantum value near this tilt~\cite{GigenaEtAl2025}. In this Letter, we show that this separation persists at every fixed level: for the symmetric doubly tilted family, the NPA error divided by the quantum violation diverges as the critical tilt is approached, while the violation itself vanishes cubically.

\emph{Correlations and quantum models.} Alice and Bob receive settings $x,y\in\{0,1\}$ and record outcomes $a,b\in\{\pm1\}$. Their correlations are conditional distributions $p(a,b|x,y)$ satisfying the no-signaling constraints: the marginals $p_A(a|x)=\sum_b p(a,b|x,y)$ and $p_B(b|y)=\sum_a p(a,b|x,y)$ are independent of the distant setting. Such a correlation is specified by four marginal averages and four joint correlators
\begin{align}
\langle A_x\rangle&=\sum_a a\,p_A(a|x),&
\langle B_y\rangle&=\sum_b b\,p_B(b|y),\nonumber\\
\langle A_xB_y\rangle&=\sum_{a,b}ab\,p(a,b|x,y).
\label{eq:correlators}
\end{align}
These eight numbers determine every no-signaling $p(a,b|x,y)$~\cite{BrunnerEtAl2014}.

A correlation is quantum if its probabilities arise from local measurements on a shared quantum state. In the tensor-product model, $p(a,b|x,y)=\Tr[\rho(E_{a|x}\otimes F_{b|y})]$, where $\rho$ is a bipartite state and $E_{a|x},F_{b|y}$ are local measurement effects. Alternatively, the effects may act on a common Hilbert space, with $p(a,b|x,y)=\Tr[\rho E_{a|x}F_{b|y}]$ and $[E_{a|x},F_{b|y}]=0$. A common local Naimark dilation permits projective measurements~\cite{Naimark1940}, and Jordan's lemma reduces finite-dimensional strategies to mixtures of two-qubit correlations~\cite{Masanes2005,PanwarPandyaWiesniak2022}. Their convex hull is compact and equals both quantum-model sets~\cite{Masanes2005,Fritz2012}. We denote this common set by $\Q$; SM Secs.~S1 and S2.1 give the compactness and model-identification arguments~\cite{Supplement}. Maximizing a Bell functional over $\Q$ therefore reduces to maximizing over two-qubit strategies. The NPA hierarchy converges to $\Q$ as its level tends to infinity~\cite{NPA2008}.

\emph{The doubly tilted CHSH inequality and its NPA bounds.} The symmetric doubly tilted CHSH expression is
\begin{align}
\beta_\alpha(p)={}&\langle A_0B_0\rangle+\langle A_0B_1\rangle
+\langle A_1B_0\rangle-\langle A_1B_1\rangle\nonumber\\
&+\alpha(\langle A_0\rangle+\langle B_0\rangle),\qquad 0\leq\alpha<1.
\label{eq:bell}
\end{align}
Its local bound is $\omega_{\rm L}(\alpha)=2+2\alpha$. We denote its maximal quantum value by $\omega_{\rm Q}(\alpha)=\max_{p\in\Q}\beta_\alpha(p)$ and its quantum violation by $\Delta_\alpha=\omega_{\rm Q}(\alpha)-\omega_{\rm L}(\alpha)$. For $\alpha=2(1-\eta)/\eta$, the same strategies maximize the CHSH violation with equal independent detection efficiencies $\eta$ and every no-click assigned $+1$~\cite{GigenaEtAl2025}. The limit $\alpha\uparrow1$ approaches the critical efficiency $\eta^*=2/3$ from above~\cite{Eberhard1993}.

For the projective measurements above, define the observables $A_x=E_{+1|x}-E_{-1|x}$ and $B_y=F_{+1|y}-F_{-1|y}$. They are self-adjoint, with $A_x^2=B_y^2=\id$ and $[A_x,B_y]=0$. Replacing the averages in Eq.~\eqref{eq:bell} by the corresponding products defines the formal polynomial
\[
h_\alpha=A_0B_0+A_0B_1+A_1B_0-A_1B_1+\alpha(A_0+B_0).
\]
Substituting a strategy's measurement operators into $h_\alpha$ gives its Bell operator $H_\alpha$; the Bell value is its expectation $\beta_\alpha(p)=\Tr(\rho H_\alpha)$. Since both models give the same correlations here, we use the commuting-operator notation $A_xB_y$ throughout; in a tensor-product realization it denotes $A_x\otimes B_y$.

The NPA hierarchy bounds this Bell value through expectations of products of observables. We call such a product $w$ a word and its expectation $\langle w\rangle$ a moment. For a finite list $w_1,\ldots,w_N$, its moment matrix has entries $\Gamma_{ij}=\langle w_i^*w_j\rangle$, where $*$ denotes the adjoint. For any coefficients $c_j$, the operator $f=\sum_jc_jw_j$ satisfies
\[
\bm c^*\Gamma\bm c=\langle f^*f\rangle\geq0.
\]
Thus every quantum realization supplies a positive semidefinite moment matrix. NPA imposes this necessary condition on candidate moments, without requiring an explicit realizing state and measurements.

At level $L$, let $\mathcal W_L$ contain the identity and all words of length at most $L$. We remove repeated representations by cancelling adjacent equal observables and commuting Alice's symbols past Bob's. For example, $\mathcal W_1=(\id,A_0,A_1,B_0,B_1)$; level $1+AB$ adds the four products $A_xB_y$. The matrix entries involve products of two such words, so they require moments through length $2L$. NPA replaces each of these unknown expectations $\langle w\rangle$ by an optimization variable $y_w$. Words that reduce to the same product share the same variable. Thus $y_{\id}=1$, while $y_{w^*}=\overline{y_w}$ preserves the adjoint relation.

To evaluate a polynomial from these variables, we replace each word by its assigned value: $\ell(\sum_wc_ww)=\sum_wc_wy_w$. We call this linear rule the truncated moment functional. In particular, $\ell(h_\alpha)$ is the candidate Bell value. We maximize it subject to moment-matrix positivity to obtain an upper bound on the quantum maximum:
\begin{equation}
\begin{aligned}
\omega_{\rm NPA}^{(L)}(\alpha)={}&\max_{\ell}\bigl\{\ell(h_\alpha):\ell(\id)=1,\\
&\quad\bigl[\ell(w_i^*w_j)\bigr]_{w_i,w_j\in\mathcal W_L}\succeq0\bigr\}.
\end{aligned}
\label{eq:level}
\end{equation}
Every quantum strategy supplies feasible moments, so $\omega_{\rm NPA}^{(L)}\geq\omega_{\rm Q}$. The maximum is attained: positivity and unit diagonal bound every matrix entry, and the constraints are closed, making the finite-dimensional feasible set compact. At level $1$ the candidate moments need not yet define nonnegative joint probabilities; level $2$ suffices for that extraction below.

Near the critical tilt the quantum and local values approach one another. The relevant question is whether a fixed NPA level resolves their shrinking separation. Our main result shows that its error exceeds every constant multiple of the scale $(1-\alpha)^3$.
\begin{theorem}\label{thm:main}
For every fixed finite $L\geq1$,
\begin{equation}
\lim_{\alpha\uparrow1}\frac{\omega_{\rm NPA}^{(L)}(\alpha)-\omega_{\rm Q}(\alpha)}{(1-\alpha)^3}=+\infty.
\label{eq:divergence}
\end{equation}
Every standard NPA relaxation defined by a fixed finite list of words in the measurement projectors admits a set of correlations strictly containing $\Q$.
\end{theorem}

The dual NPA program yields a sum-of-squares certificate~\cite{PironioEtAl2010,Tavakoli2024,IoannouRosset2021}: on any quantum strategy, the difference between the bound and the Bell value is a sum of squared norms. We construct strategies whose gap from the quantum bound, divided by the vanishing quantum violation, tends to the Motzkin polynomial. This polynomial is nonnegative~\cite{Motzkin1967}, but even adding a nonnegative constant cannot make it a sum of polynomial squares. An NPA error bounded by a constant multiple of the quantum violation would force such a representation. To derive the contradiction, we first calculate the quantum violation and then construct the strategies on that scale.

\emph{The cubic quantum violation.} Theorem~1 of Ref.~\cite{GigenaEtAl2025} gives $\omega_{\rm Q}(\alpha)$ as the largest root of a quartic $F(\lambda,\alpha)$, where $\lambda$ is the Bell-eigenvalue variable (SM Sec.~S4). We find the gap above the local bound by comparing the polynomial at these two Bell values. Setting $T=1-\alpha$, direct substitution gives
\[
\begin{gathered}
F(4-2T,1-T)=-4T^3(3-T)^4,\\
F(\omega_{\rm Q}(1-T),1-T)=0,\\
\left.\frac{\partial F}{\partial\lambda}\right|_{(\lambda,\alpha)=(4,1)}=243.
\end{gathered}
\]
Thus, at fixed tilt $1-T$, the polynomial changes from a small negative number to zero as $\lambda$ moves from the local bound $4-2T$ to the quantum bound. The distance between these values is precisely the quantum violation $\Delta_{1-T}$.

At $T=0$ both bounds equal $4$. Changing the tilt by $T$ changes any Bell value by at most $2T$, because each marginal average has magnitude at most one. Both endpoints therefore remain within $O(T)$ of $4$, and the derivative throughout the interval is $243+O(T)$. The mean-value theorem equates the change in the polynomial to its slope at an intermediate point times the distance between the bounds:
\[
[243+O(T)]\Delta_{1-T}=4T^3(3-T)^4.
\]
The leading coefficient on the right is $4\cdot3^4=324$. Dividing by the limiting slope gives $324/243=4/3$, hence $\Delta_{1-T}=\tfrac43T^3+O(T^4)$. Adding the local bound yields
\begin{equation}
\omega_{\rm Q}(1-T)=4-2T+\tfrac43T^3+O(T^4).
\label{eq:quantum-series}
\end{equation}
The quantum violation therefore vanishes cubically. The optimal-angle formula in Ref.~\cite{GigenaEtAl2025} gives $\theta_*(T)\sim2\sqrt T$ for the angle between each party's observables, while the optimal state approaches a product state (SM Sec.~S4). These angle and violation scales determine the family of strategies used below.

\emph{The Motzkin limit.} Writing $T=t^2$ makes the angles of order $t$ and the quantum violation of order $t^6$. With $X,Z$ the Pauli matrices, take
\begin{align}
A_0&=Z,& A_1&=\cos(tu)Z+\sin(tu)X,\nonumber\\
B_0&=Z,& B_1&=\cos(tv)Z+\sin(tv)X,
\label{eq:scaled-observables}
\end{align}
and write the resulting Bell operator as $H(t;u,v)$: it is $h_\alpha$ evaluated at $\alpha=1-t^2$ and on these measurements. Here $u,v\in\mathbb R$ vary independently; we write $H(t)$ when their values are fixed. An NPA certificate holds for every such pair, so we seek states with an order-$t^6$ gap throughout this two-parameter family.

Substituting the measurements gives the operator explicitly. With $c_u=\cos(tu)$, $s_u=\sin(tu)$, and similarly for $v$, and writing $ZI=Z\otimes I$, etc.,
\[
\begin{aligned}
H(t)={}&(1-t^2)(ZI+IZ)\\
&+[2-(1-c_u)(1-c_v)]ZZ\\
&+(1-c_u)s_vZX+(1-c_v)s_uXZ-s_us_vXX.
\end{aligned}
\]
At $t=0$ this becomes $2ZZ+ZI+IZ=\operatorname{diag}(4,-2,-2,0)$ in the computational basis. Its top eigenvector is $|00\rangle$, with eigenvalue gaps $6,6,4$ to the other basis states. For small $t$, the off-diagonal entries couple it to the other basis states:
\[
\begin{aligned}
\langle01|H(t;u,v)|00\rangle&=\tfrac12t^3u^2v+O(t^5),\\
\langle10|H(t;u,v)|00\rangle&=\tfrac12t^3uv^2+O(t^5),\\
\langle11|H(t;u,v)|00\rangle&=-t^2uv+O(t^4).
\end{aligned}
\]
The last entry, for example, is exactly $-\sin(tu)\sin(tv)$. A real correction $z$ along $|11\rangle$ contributes $4z^2+2t^2uvz$ at leading order to the gap, so minimizing it requires $z$ of order $t^2$. The other two couplings similarly require corrections of order $t^3$.

We therefore take a real vector $\zeta_t=(1,at^3,bt^3,ct^2)^T$ and determine its coefficients by minimizing the leading terms of the gap. Its expectation of $\omega_{\rm Q}(1-t^2)I-H(t;u,v)$ starts with $t^4(4c+uv)^2/4$. Since this term would dominate a gap of order $t^6$, it must vanish, giving $c=-uv/4$. The coefficient of $t^6$ is then
\begin{align}
&6\left(a-\frac{u^2v}{12}\right)^2+6\left(b-\frac{uv^2}{12}\right)^2\nonumber\\
&\qquad+\frac43+\frac{u^4v^2+u^2v^4}{48}-\frac{u^2v^2}{4}.
\label{eq:completion}
\end{align}
Choosing $a=u^2v/12$ and $b=uv^2/12$ cancels the two squares. This gives the test vector
\begin{equation}
\xi_t=\left(1,\frac{t^3u^2v}{12},\frac{t^3uv^2}{12},-\frac{t^2uv}{4}\right)^T.
\label{eq:xi}
\end{equation}
With these choices, expanding the sines and cosines gives the leading gap. Their remainders are uniform on each bounded range of $u,v$, yielding (SM Sec.~S6)
\begin{equation}
t^{-6}\langle\xi_t,[\omega_{\rm Q}(1-t^2)I-H(t)]\xi_t\rangle
\xrightarrow[t\downarrow0]{}\tfrac43M(u/2,v/2),
\label{eq:motzkin-limit}
\end{equation}
locally uniformly, where $M(U,V)=U^4V^2+U^2V^4-3U^2V^2+1$ is the Motzkin polynomial~\cite{Motzkin1967}. Its nonnegativity follows from the arithmetic--geometric mean inequality applied to $U^4V^2,U^2V^4,1$.

For the physical state $\psi_t=\xi_t/\|\xi_t\|$, the Bell value is $\langle\psi_t,H(t)\psi_t\rangle$. Since $\|\xi_t\|^2=1+O(t^4)$ and $\Delta_{1-t^2}\sim4t^6/3$, normalization and division by the quantum violation turn Eq.~\eqref{eq:motzkin-limit} into
\begin{equation}
\frac{\omega_{\rm Q}(1-t^2)-\langle\psi_t,H(t)\psi_t\rangle}
{\Delta_{1-t^2}}
\xrightarrow[t\downarrow0]{}M(u/2,v/2).
\label{eq:strategy-gap}
\end{equation}
If either angle is zero, the state is $|00\rangle$, its Bell value equals the local bound, and the ratio is exactly $1$. At $u=v=2$ it tends to zero, agreeing with the optimal angular scale. Although the optimal angles are equal, an NPA certificate must hold for independently varied angles. This will force a polynomial identity in both variables.

\emph{The SOS obstruction.} To relate these gaps to the NPA error, we fix a level $L$ and write the words in $\mathcal W_L$ as $\bm w=(w_1,\ldots,w_N)^T$. An identity $b\id-h_\alpha=\sum_j f_j^*f_j$ certifies the upper bound $b$, because its expectation on any strategy gives $b-\beta_\alpha(p)\geq0$. Semidefinite duality supplies such an identity at the NPA optimum: the moment assignment $\ell(\id)=1$, $\ell(w)=0$ for all other reduced words gives the positive definite identity matrix, satisfying the strict-feasibility condition (SM Sec.~S5). Thus, after reducing words by the measurement relations,

\begin{equation}
\omega_{\rm NPA}^{(L)}(\alpha)\id-h_\alpha=\bm w^*Q\bm w,\qquad Q\succeq0.
\label{eq:dual}
\end{equation}
Factoring $Q=C^*C$ gives $f_j=\sum_kC_{jk}w_k$. On a vector $\xi$, the right side is therefore $\sum_j\|f_j\xi\|^2$. The words remain fixed as $\alpha$ changes, but their coefficients $C_{jk}$ may vary. We need a uniform bound on these coefficients to control the limit of the evaluated squares.

The required bound follows by comparing the coefficients of the identity word in Eq.~\eqref{eq:dual}. Because $\mathcal W_L$ contains one representative of each reduced word, $w_i^*w_j$ reduces to the identity exactly when $i=j$. Since $h_\alpha$ has no identity term, this comparison gives
\begin{equation}
\Tr Q=\sum_iQ_{ii}=\omega_{\rm NPA}^{(L)}(\alpha)\leq4+2\alpha\leq6.
\label{eq:trace}
\end{equation}
The upper bound follows because each observable moment entering the Bell expression has modulus at most one. Taking $C=Q^{1/2}$, we have $\|C\|\leq\sqrt6$. In particular, every coefficient $C_{jk}$ stays bounded independently of the tilt.

To obtain a contradiction, assume that Eq.~\eqref{eq:divergence} is false. This means that the nonnegative NPA error divided by $t^6$ stays bounded along some sequence $t_n\downarrow0$. We may choose a subsequence along which these ratios converge to a number $r\geq0$, and write
\begin{equation}
\omega_{\rm NPA}^{(L)}(1-t_n^2)-\omega_{\rm Q}(1-t_n^2)=r_nt_n^6,
\qquad r_n\to r.
\label{eq:bounded-excess}
\end{equation}
At each tilt $1-t_n^2$, let $Q_n$ be an optimal certificate and $C_n=Q_n^{1/2}$. Evaluating its factors on the measurements in Eq.~\eqref{eq:scaled-observables} and applying them to $\xi_t$ gives $N$ four-component vectors. Collecting them in $g_n(t,u,v)$ turns their sum of squared norms into $\|g_n\|^2$. The certificate at $t=t_n$ then gives
\[
\|g_n(t_n,u,v)\|^2
=\langle\xi_{t_n},[\omega_{\rm NPA}^{(L)}(1-t_n^2)I-H(t_n)]\xi_{t_n}\rangle.
\]
The right side separates into the test-vector gap in Eq.~\eqref{eq:motzkin-limit} and the NPA error times $\|\xi_{t_n}\|^2$. Dividing the factors by $t_n^3$ divides their squared norm by $t_n^6$, giving
\begin{align}
v_n(u,v)&=t_n^{-3}g_n(t_n,u,v),\nonumber\\
\|v_n(u,v)\|^2&\longrightarrow\tfrac43M(u/2,v/2)+r.
\label{eq:certificate-limit}
\end{align}
On every fixed bounded range of $u,v$, the difference between the two sides tends to zero uniformly. The additional term is $r$ because the error contribution is $r_n\|\xi_{t_n}\|^2$ and $\|\xi_{t_n}\|^2\to1$.

To obtain polynomial squares from this limit, we approximate the trigonometric factors by Taylor polynomials before taking the limit. Their coefficients are polynomials in the angles: for example, $\sin(tu)=tu-t^3u^3/6+O(t^5)$ and $\cos(tu)=1-t^2u^2/2+O(t^4)$. For each fixed $n$, expanding $g_n(t,u,v)$ in $t$ while holding $C_n$ fixed gives coefficient vectors $g_{n,j}(u,v)$ at order $t^j$. Each component of $g_{n,j}$ has degree at most $j$: the Taylor coefficients of both the measurements and $\xi_t$ have angle degree no greater than their order in $t$. We approximate the rescaled vector $v_n$ by retaining terms through $t^3$:
\begin{equation}
p_n(u,v)=t_n^{-3}\sum_{j=0}^3t_n^jg_{n,j}(u,v).
\label{eq:polynomial-factors}
\end{equation}
This truncation leaves terms of order $t_n^4$ or higher before rescaling. Their error bound can be chosen independently of $n$: the word list is fixed, and multiplying by $C_n$ enlarges the error by at most $\sqrt6$. After division by $t_n^3$, the error is therefore at most a constant times $t_n$, throughout any fixed bounded range of $u,v$. Thus $p_n-v_n\to0$, and replacing $v_n$ by $p_n$ leaves the squared-norm limit in Eq.~\eqref{eq:certificate-limit} unchanged.

The two bounds now play different roles. Equation~\eqref{eq:trace} controls the Taylor remainder, while Eq.~\eqref{eq:certificate-limit}, together with $p_n-v_n\to0$, bounds the polynomial factors on $[-1,1]^2$. Each component of $p_n$ has total degree at most three, and hence degree at most three in either variable. Its values at four fixed distinct choices of $u$ and four of $v$ therefore determine its coefficients. The interpolation formula expresses each coefficient as a fixed linear combination of these sixteen values. Hence the coefficients of every component of $p_n$ are bounded. There are only finitely many of them at fixed $L$, so we may choose a subsequence along which they all converge. Let $q(u,v)$ be the resulting polynomial vector. Coefficient convergence at fixed degree gives convergence on every bounded range of the variables, and Eq.~\eqref{eq:certificate-limit} becomes
\begin{equation}
\|q(u,v)\|^2=\tfrac43M(u/2,v/2)+r.
\label{eq:polynomial-identity}
\end{equation}
For each component, $|q_j|^2=(\operatorname{Re}q_j)^2+(\operatorname{Im}q_j)^2$, so the left side is a sum of squares of real polynomials. SM Sec.~S8 gives the full estimates.

Substituting $u=2U,v=2V$ in Eq.~\eqref{eq:polynomial-identity} and multiplying by $3/4$ would give $M(U,V)+c=\sum_j\varphi_j(U,V)^2$ for some real polynomials $\varphi_j$ and $c=3r/4\geq0$. Each $\varphi_j$ has degree at most three: squares of highest-degree homogeneous terms cannot cancel. On either axis the sum equals the constant $1+c$. Since every square is nonnegative, each restricted $\varphi_j$ is bounded on the whole real line and is therefore constant. Apart from the constant, the remaining monomials of degree at most three contain both variables. Hence
\[
\varphi_j(U,V)=a_j+b_jUV+d_jU^2V+e_jUV^2.
\]
Only $(b_jUV)^2$ contributes to $U^2V^2$ in the sum, giving coefficient $\sum_jb_j^2\geq0$, whereas $M+c$ has coefficient $-3$. This contradiction proves Eq.~\eqref{eq:divergence}. SM Sec.~S7 gives an alternative proof of the monomial restriction using the Newton-polytope theorem~\cite{Reznick1978}.

\emph{Nonquantum correlations in finite NPA relaxations.} Equation~\eqref{eq:divergence} gives feasible moments with Bell values above the quantum maximum. To obtain nonquantum correlations in an arbitrary finite NPA relaxation, these moments must give valid probabilities and satisfy that relaxation's moment constraints. Let its projector-word list be $\mathcal S$. Substituting $E_{a|x}=(\id+aA_x)/2$ and $F_{b|y}=(\id+bB_y)/2$ expresses each of its words as a linear combination of words in $\mathcal W_d$, for some finite $d\geq2$. The observable relations imply the projector relations. Writing the $j$th projector word as $s_j=\sum_iR_{ij}w_i$, the corresponding moment matrices obey $\Gamma_{\mathcal S}=R^*\Gamma_dR$. Positivity is preserved because $z^*\Gamma_{\mathcal S}z=(Rz)^*\Gamma_d(Rz)\geq0$.

At $d\geq2$, the candidate joint probability $p(a,b|x,y)=\ell(E_{a|x}F_{b|y})$ is nonnegative: write $P=E_{a|x}F_{b|y}=\sum_i\sigma_iw_i$ and let $\boldsymbol\sigma=(\sigma_i)_i$ be its coefficient vector in $\mathcal W_d$. Since $P$ is a projector, $\ell(P)=\ell(P^*P)=\boldsymbol\sigma^*\Gamma_d\boldsymbol\sigma\geq0$. Completeness gives normalization and no-signaling. An optimizer of the level-$d$ bound therefore supplies a valid correlation feasible for $\mathcal S$. For $\alpha$ sufficiently close to one, its Bell value exceeds $\omega_{\rm Q}(\alpha)$ by Eq.~\eqref{eq:divergence}; no quantum correlation can attain that value. Thus every fixed finite projector-word relaxation contains correlations outside $\Q$.

Since $\Delta_\alpha\sim4(1-\alpha)^3/3$, the NPA error also exceeds every constant multiple of the quantum violation:
\begin{equation}
\lim_{\alpha\uparrow1}
\frac{\omega_{\rm NPA}^{(L)}(\alpha)-\omega_{\rm Q}(\alpha)}
{\omega_{\rm Q}(\alpha)-\omega_{\rm L}(\alpha)}=+\infty.
\label{eq:relative}
\end{equation}
This diverging ratio lets us mix the nonquantum correlations with an increasingly large fraction of a local one without losing their Bell violation above the quantum bound. Let $p_{d,\alpha}$ attain $\omega_{\rm NPA}^{(d)}$, put $e_{d,\alpha}=\omega_{\rm NPA}^{(d)}-\omega_{\rm Q}$, and let $p_{\rm det}$ always output $a=b=1$. The mixture $q_\alpha=\lambda_\alpha p_{d,\alpha}+(1-\lambda_\alpha)p_{\rm det}$ has Bell value $\omega_{\rm L}+\lambda_\alpha(\Delta_\alpha+e_{d,\alpha})$. Choosing $\lambda_\alpha=2\Delta_\alpha/(\Delta_\alpha+e_{d,\alpha})$ makes this value $\omega_{\rm Q}+\Delta_\alpha$. Equation~\eqref{eq:relative} gives $\lambda_\alpha\to0$, so the weights lie in $(0,1)$ for tilts sufficiently close to one. The mixtures remain feasible and converge to $p_{\rm det}$ while violating the quantum bound (SM Sec.~S10). These conclusions include the almost quantum set, defined by the $1+AB$ projector list~\cite{NGHA2015}.

\emph{Consequences and open questions.} For each fixed NPA level, nonquantum feasible correlations occur arbitrarily close to a local deterministic correlation. Here $\Q$ is compact and every Bell maximum is attained on two qubits, so the failure of finite-level certification does not require nonclosure or unbounded physical dimension. In the one-sided tilted family the optimal state also approaches a product state, while the quantum bound remains certified at $1+AB$ throughout the limit~\cite{BampsPironio2015}.

For the doubly tilted family, each fixed level fails for every tilt sufficiently close to one. This answers the question in Ref.~\cite{GigenaEtAl2025} of whether one level suffices for the whole family. The interval of failure may shrink as the level increases, leaving open whether one fixed Bell functional in $(2,2,2)$ is missed at every finite level. Such a functional would give one Bell inequality whose exact quantum bound has no finite NPA certificate. The polynomial obstruction here emerges as the tilt varies, so it does not establish this stronger result. A fixed game with a strict gap at every finite NPA level is known in a larger scenario~\cite{FanizzaEtAl2025}.

Failure of a finite NPA description still allows membership in $\Q$ to be decided by other means. Its eight correlation coordinates and Carathéodory's theorem~\cite{Barvinok2002} give a mixture of at most nine two-qubit correlations. The states, measurements and mixture weights then form a finite set of unknown real numbers subject to polynomial constraints. Given probabilities specified as ratios of integers or exact real roots of integer polynomials, real quantifier elimination decides whether these constraints have a solution. With a distance threshold specified in the same way, this method also decides whether the distribution lies within that Euclidean distance of $\Q$ (SM Sec.~S12). These procedures need not be efficient, but they show why undecidability for unrestricted nonlocal games~\cite{MIPstarRE} does not extend to this fixed scenario.

One may also strengthen the NPA constraints with conditions expressing optimality of the state or measurements~\cite{AraujoEtAl2026}. Such conditions change the form of the dual certificate, so determining whether they remove this family's finite-level gap requires a further argument. More generally, the polynomial description above does not ensure a finite semidefinite representation~\cite{Scheiderer2018}. This leaves the broader question~\cite{FarkasQLunch2026}: Does the complete quantum set in the simplest Bell scenario admit an exact finite semidefinite representation?

\emph{Note added.} This work was completed independently and submitted to Physical Review Letters on 14 July 2026. Pakhunov's preprints, submitted to arXiv on 15 July 2026, claim to prove finite-level nonexactness for the same doubly tilted family using feasible NPA moments~\cite{Pakhunov2026a} and to identify the Motzkin limit~\cite{Pakhunov2026b}. The finite-level argument is computer-assisted, with exact-arithmetic verifications forming part of the proof as described in the preprint; we have not independently verified those computations, and none of our arguments relies on either preprint. By contrast, our analytical proof works on the dual side: rather than constructing feasible NPA moment curves, we evaluate optimal NPA certificates on explicit two-qubit strategies.

\begin{acknowledgments}
We thank Mateus Ara\'ujo for valuable feedback.
We acknowledge support from the KLAR Grant No. BNI/PST/2023/1/00013/U/00001, funded by NAWA.
ChatGPT was used to assist with the presentation of the author's analytical argument for the initial submission. The author subsequently rewrote the main text in response to the referees' comments, independently verified all proofs and results, and takes full responsibility for the scientific content.
\end{acknowledgments}
\paragraph*{Data availability.}
All analytical derivations needed to reproduce the results are provided in the Supplemental Material.

\end{document}

\title{Supplemental Material for\\ No Finite NPA Level Characterizes the Complete Quantum Set in the Simplest Bell Scenario}
\author{Anubhav Chaturvedi}
\affiliation{Division of Quantum Optics and Information, Institute of Theoretical Physics and Astrophysics, Faculty of Mathematics, Physics and Informatics, University of Gda\'nsk, 80-308 Gda\'nsk, Poland}
\maketitle
\setcounter{secnumdepth}{3}
\renewcommand{\thesection}{S\arabic{section}}
\renewcommand{\thesubsection}{\thesection.\arabic{subsection}}
\makeatletter
\def\p@subsection{}
\makeatother
\renewcommand{\theequation}{S\arabic{equation}}
\renewcommand{\thetheorem}{S\arabic{theorem}}
We construct two-qubit states and measurements for which the gap between the exact quantum bound and the attained Bell value, divided by the maximal quantum violation above the local bound, tends to the Motzkin polynomial. A fixed NPA level with bounded relative error would give a polynomial sum-of-squares representation of this polynomial plus a nonnegative constant, which Section~\ref{supp:sec:not-sos} rules out.

Sections~\ref{supp:sec:bell} and \ref{supp:sec:jordan} give the setup and two-qubit reduction. Section~\ref{supp:sec:cubic-scale} gives the quadratic form used to calculate the test-state gap; Section~\ref{supp:sec:endpoint} derives the endpoint scales and higher-order terms from the quantum bound and optimal measurements of Ref.~\cite{GigenaEtAl2025}. Sections~\ref{supp:sec:npa}--\ref{supp:sec:main-proof} prove the NPA divergence, and Sections~\ref{supp:sec:transfer}--\ref{supp:sec:scope} give its correlation-set consequences and scope.

All remainder bounds are locally uniform in the angular coordinates: $R(t,x,y)=O_K(t^m)$ means that for every compact $K\subset\R^2$ there are $C_K,\delta_K>0$ such that
\begin{equation}
\sup_{(x,y)\in K}|R(t,x,y)|\leq C_K|t|^m
\qquad (0<|t|<\delta_K).
\label{supp:eq:uniform-big-O-definition}
\end{equation}
For vector or matrix remainders, the same definition uses the Euclidean or operator norm. We omit $K$ when no angular variables are present.

\section{Bell correlations, the doubly tilted CHSH functional, and quantum strategies}
\label{supp:sec:bell}

\subsection{Correlations and Bell values}

Alice and Bob choose settings $x,y\in\{0,1\}$ and obtain outcomes $a,b\in\{\pm1\}$. Their correlation $p(a,b|x,y)$ consists of normalized, nonnegative conditional probabilities with no-signaling marginals
\begin{equation}
p_A(a|x)=\sum_b p(a,b|x,y),
\qquad
p_B(b|y)=\sum_a p(a,b|x,y).
\label{supp:eq:nosig-a}
\end{equation}
No-signaling means that the first marginal is independent of $y$ and the second of $x$. The corresponding marginal expectations and joint correlators are
\begin{equation}
\langle A_x\rangle=\sum_a a\,p_A(a|x),
\qquad
\langle B_y\rangle=\sum_b b\,p_B(b|y),
\qquad
\langle A_xB_y\rangle=\sum_{a,b}ab\,p(a,b|x,y).
\label{supp:eq:correlator}
\end{equation}
These eight real numbers specify the full correlation through~\cite{BrunnerEtAl2014}
\begin{equation}
p(a,b|x,y)=\frac14\left[1+a\langle A_x\rangle+b\langle B_y\rangle+ab\langle A_xB_y\rangle\right],
\label{supp:eq:behavior-reconstruction}
\end{equation}
subject to nonnegativity of the reconstructed probabilities.

For $0\leq\alpha<1$, the symmetric doubly tilted CHSH functional is
\begin{equation}
\beta_\alpha(p)=\langle A_0B_0\rangle+\langle A_0B_1\rangle
+\langle A_1B_0\rangle-\langle A_1B_1\rangle
+\alpha(\langle A_0\rangle+\langle B_0\rangle).
\label{supp:eq:bell-functional-on-behaviors}
\end{equation}
We use $h_\alpha$ for the corresponding polynomial in observable symbols,
\begin{equation}
h_\alpha=A_0B_0+A_0B_1+A_1B_0-A_1B_1+\alpha(A_0+B_0).
\label{supp:eq:bell-functional}
\end{equation}
A choice of observables turns $h_\alpha$ into a Bell operator whose expectation is $\beta_\alpha(p)$. The CHSH local bound is $2$~\cite{CHSH1969}, the two marginal terms contribute at most $2\alpha$, and the deterministic assignment $A_0=A_1=B_0=B_1=1$ attains both bounds. Hence
\begin{equation}
\omega_{\rm L}(\alpha)=2+2\alpha.
\label{supp:eq:local-value}
\end{equation}

The same functional describes CHSH with equal detector efficiencies and outcome $+1$ assigned to every no-detection event~\cite{MiklinEtAl2022,GigenaEtAl2025}. Suppose each detector fires with probability $\eta$, independently of the other detector, the settings, and the ideal outcomes. For an ideal correlation $p$,
\begin{align*}
\langle A_xB_y\rangle_\eta
&=\eta^2\langle A_xB_y\rangle
+\eta(1-\eta)(\langle A_x\rangle+\langle B_y\rangle)+(1-\eta)^2,\\
\mathrm{CHSH}_\eta(p)
&=\eta^2\beta_\alpha(p)+2(1-\eta)^2,
\qquad \alpha=\frac{2(1-\eta)}{\eta}.
\end{align*}
The setting-$1$ marginals cancel in the CHSH sum, while the setting-$0$ marginals each enter twice. Thus $\alpha\uparrow1$ corresponds to approaching the critical efficiency $\eta^*=2/3$ from above~\cite{Eberhard1993}. The local bound becomes $2$, and both the violation above the local bound and the NPA error acquire the factor $\eta^2$, preserving their ratio.

\subsection{CHSH and one-sided tilted CHSH}

Short lists of observable products already recover several exact quantum bounds. These lists index the NPA moment matrices defined in Section~\ref{supp:sec:npa}. Observable level $1$ uses
\begin{equation}
\mathcal W_1=\{\id,A_0,A_1,B_0,B_1\},
\label{supp:eq:level-one-list}
\end{equation}
and gives the exact CHSH maximum $2\sqrt2$~\cite{Tsirelson1980,NPA2007}. The intermediate level $1+AB$ uses
\begin{equation}
\mathcal W_{1+AB}=\mathcal W_1\cup\{A_xB_y:x,y\in\{0,1\}\}.
\label{supp:eq:level-one-ab-list}
\end{equation}
For the one-sided tilted CHSH operator~\cite{AMP2012},
\begin{equation}
I_\gamma=A_0B_0+A_0B_1+A_1B_0-A_1B_1+\gamma A_0,
\qquad 0\leq\gamma<2,
\label{supp:eq:one-sided-tilted}
\end{equation}
tight sum-of-squares decompositions use words in $\mathcal W_{1+AB}$ and self-test the maximizing strategy~\cite{BampsPironio2015,SupicBowles2020}. Consequently, level $1+AB$ is exact when either marginal tilt vanishes~\cite{GigenaEtAl2025}.

\subsection{Quantum strategies and projective measurements}

A quantum correlation has the form
\begin{equation}
p(a,b|x,y)=\Tr\!\left[\rho_{AB}\,E_{a|x}\otimes F_{b|y}\right],
\label{supp:eq:quantum-behavior}
\end{equation}
where $\rho_{AB}$ is a density operator on $\mathcal H_A\otimes\mathcal H_B$ and
\begin{equation}
E_{a|x}\succeq0,\quad\sum_aE_{a|x}=I_A,
\qquad
F_{b|y}\succeq0,\quad\sum_bF_{b|y}=I_B.
\label{supp:eq:povm-relations}
\end{equation}
Write $\Q$ for this tensor-product correlation set and $\omega_{\rm Q}(\alpha)=\sup_{p\in\Q}\beta_\alpha(p)$. Section~\ref{supp:subsec:models} establishes its equality with the finite-dimensional, closure, and commuting-operator versions.

By Naimark dilation, the local POVMs can be replaced by projective measurements on larger local spaces without changing any joint probability~\cite{Naimark1940}. The dilation is common to both settings at each party and preserves finite dimensionality. Projective measurements also suffice in the commuting model by Proposition~3.4(b) of Ref.~\cite{Fritz2012}. Denote the resulting local projectors by $P_{a|x}$ and $Q_{b|y}$. Taking the difference of the two outcome projectors gives the binary observables
\begin{equation}
A_x=P_{+|x}-P_{-|x},
\qquad B_y=Q_{+|y}-Q_{-|y}
\label{supp:eq:reflections}
\end{equation}
which satisfy
\begin{equation}
A_x=A_x^*,\quad A_x^2=I,
\qquad B_y=B_y^*,\quad B_y^2=I,
\label{supp:eq:reflection-relations}
\end{equation}
with inverse relations
\begin{equation}
P_{a|x}=\frac{I+aA_x}{2},
\qquad Q_{b|y}=\frac{I+bB_y}{2}.
\label{supp:eq:projector-reflection}
\end{equation}

\section{Jordan's lemma, quantum models, and the two-qubit Bell operator}
\label{supp:sec:jordan}

Jordan's lemma decomposes any finite-dimensional pair of binary projective observables into common orthogonal invariant subspaces of dimension at most two~\cite{Halmos1969,Masanes2005,PanwarPandyaWiesniak2022}. Applied locally, it gives $\mathcal H_A=\bigoplus_iV_i$ and $\mathcal H_B=\bigoplus_jW_j$, with all measurement projectors preserving each $V_i\otimes W_j$.

\subsection{Equality of the quantum models in the \texorpdfstring{$(2,2,2)$}{(2,2,2)} scenario}
\label{supp:subsec:models}

Let $\Pi_i^A,\Pi_j^B$ project onto the local Jordan blocks, and put $P_{ij}=\Pi_i^A\otimes\Pi_j^B$. The block weights are $r_{ij}=\Tr(P_{ij}\rho_{AB})$; for $r_{ij}>0$, the normalized block state is $\rho_{ij}=P_{ij}\rho_{AB}P_{ij}/r_{ij}$. Since the measurements preserve these blocks, their probabilities satisfy $p=\sum_{i,j:r_{ij}>0}r_{ij}p_{ij}$, where $p_{ij}$ is generated by $\rho_{ij}$ and the restricted measurements~\cite[Eqs.~(12)--(15)]{Masanes2005}. Let $\mathcal K_2$ contain the correlations generated by two-qubit density operators and two binary projective measurements per party, including trivial projectors. Conversely, any finite mixture of these correlations is realized by correlated classical labels selecting the corresponding local direct-sum blocks. The block decomposition and the common local dilation therefore give
\begin{equation}
\mathcal Q_{\rm fd}=\operatorname{conv}\mathcal K_2.
\label{supp:eq:finite-convex-qubit-model}
\end{equation}
The eight-dimensional correlation space and Carath\'eodory's theorem~\cite{Barvinok2002} limit each mixture to nine components. The parameters defining $\mathcal K_2$ form compact sets, and the Born probabilities depend continuously on them. The mixtures therefore form the continuous image of $\Delta_8\times\mathcal K_2^9$, where $\Delta_8$ is the probability simplex on nine points. This makes $\operatorname{conv}\mathcal K_2$ compact. The usual finite-dimensional approximation of tensor-product correlations consequently introduces no additional points~\cite{ScholzWerner2008}.

For the commuting model, the local algebra generated by two binary projective measurements is nuclear; Sec.~4, Example~7(2), of Ref.~\cite{ScholzWerner2008} gives this two-projection case. Proposition~3.4 and Remark~3.6 of Ref.~\cite{Fritz2012} equate the commuting set with the tensor-product closure. The compactness above makes that closure equal to the same finite-dimensional set. Therefore
\begin{equation}
\mathcal Q_{\rm fd}=\Q=\overline{\mathcal Q_{\rm fd}}
=\mathcal Q_{\rm c}=\operatorname{conv}\mathcal K_2.
\label{supp:eq:quantum-model-equality}
\end{equation}
Theorem~8 of Ref.~\cite{NPA2008} states that a correlation belongs to the commuting set exactly when it satisfies every NPA level. Hence the hierarchy converges to $\Q$ in this scenario.

\subsection{The two-qubit Bell operator}

On each nontrivial two-dimensional Jordan block, diagonalizing the first observable gives $Z$. A phase change of the second basis vector preserves $Z$ and makes the second observable's off-diagonal entry real and nonnegative. Applied independently at Alice and Bob, these choices give
\begin{align}
A_0&=Z,\qquad A_1=cZ+rX,\qquad r=\sqrt{1-c^2},
\label{supp:eq:alice-block}\\
B_0&=Z,\qquad B_1=dZ+sX,\qquad s=\sqrt{1-d^2}.
\label{supp:eq:bob-block}
\end{align}
Here $X,Z$ are Pauli matrices and $c,d\in(-1,1)$. Allowing the endpoints $c,d=\pm1$ also includes all one-dimensional blocks: the assignment $(A_0,A_1)=(\epsilon_0,\epsilon_1)$ occurs on the $Z$ eigenspace with eigenvalue $\epsilon_0$ at $c=\epsilon_0\epsilon_1$, and similarly for Bob.

Set
\begin{equation}
\alpha=1-T,\qquad0<T\leq1.
\label{supp:eq:T-definition}
\end{equation}
The two-qubit evaluation of $h_{1-T}$ is
\begin{align}
H_T(c,d)
={}&(1+c+d-cd)Z\otimes Z+(1-c)s\,Z\otimes X
\nonumber\\
&+(1-d)r\,X\otimes Z-rs\,X\otimes X
+(1-T)(Z\otimes I+I\otimes Z).
\label{supp:eq:operator-decomposition}
\end{align}
Writing
\begin{equation}
\Omega=1+c+d-cd,
\label{supp:eq:omega-block}
\end{equation}
the matrix in the basis $|00\rangle,|01\rangle,|10\rangle,|11\rangle$ is
\begin{equation}
H_T(c,d)=
\begin{pmatrix}
\Omega+2(1-T)&(1-c)s&(1-d)r&-rs\\
(1-c)s&-\Omega&-rs&(d-1)r\\
(1-d)r&-rs&-\Omega&(c-1)s\\
-rs&(d-1)r&(c-1)s&\Omega-2(1-T)
\end{pmatrix}.
\label{supp:eq:block-matrix}
\end{equation}
Every finite-dimensional Bell operator is a direct sum of these blocks or their invariant restrictions. Each block maximum is attained by its largest-eigenvalue eigenvector, and Eq.~\eqref{supp:eq:quantum-model-equality} covers all quantum models. Therefore
\begin{equation}
\omega_{\rm Q}(1-T)=\max_{c,d\in[-1,1]}\lambda_{\max}H_T(c,d).
\label{supp:eq:block-maximum}
\end{equation}
The maximum exists by continuity on the compact parameter square and is attained on two qubits.

\section{The Bell-value excess as a quadratic form}
\label{supp:sec:cubic-scale}

To calculate the Motzkin limit in Section~\ref{supp:sec:motzkin}, we will
expand the Bell value of vectors close to $|00\rangle$. The following
identity gives their excess above the local bound directly from the Bell
matrix in Eq.~\eqref{supp:eq:block-matrix}. Denote the local bound by
\begin{equation}
\lambda_0(T)=4-2T,
\label{supp:eq:lambda-zero}
\end{equation}
and set $u=1-c$, $v=1-d$, $r=\sqrt{1-c^2}$, and $s=\sqrt{1-d^2}$. The variables $u,v$ vanish when each party's two observables coincide.
For a real vector $\eta=(1,b_1,b_2,b_3)^T$, multiplying the matrix by
$\eta$ gives
\begin{align}
\mathcal E_T(\eta)
:={}&\langle\eta,[H_T-\lambda_0(T)I]\eta\rangle
\nonumber\\
={}&-uv+2us\,b_1+2vr\,b_2-2rs\,b_3
\nonumber\\
&+(-6+2T+uv)(b_1^2+b_2^2)+(-4+4T-uv)b_3^2
\nonumber\\
&-2rs\,b_1b_2-2vr\,b_1b_3-2us\,b_2b_3.
\label{supp:eq:excess-quadratic-form}
\end{align}
The normalized vector therefore attains the Bell value
$\lambda_0(T)+\mathcal E_T(\eta)/\|\eta\|^2$.
The same identity holds for signed $r,s$ satisfying $r^2=1-c^2$ and
$s^2=1-d^2$, as used for the freely varying angles in
Section~\ref{supp:sec:motzkin}.

The corresponding unnormalized quantum gap is
\[
\langle\eta,[\omega_{\rm Q}(1-T)I-H_T]\eta\rangle
=[\omega_{\rm Q}(1-T)-\lambda_0(T)]\|\eta\|^2-\mathcal E_T(\eta).
\]

\section{Endpoint expansions from the exact quantum bound}
\label{supp:sec:endpoint}

We use the exact quantum bound and maximizing measurements from
Theorem~1 of Ref.~\cite{GigenaEtAl2025} to find their behavior near the
critical tilt. For each $0\leq\alpha<1$, that theorem identifies
$\omega_{\rm Q}(\alpha)$ with the largest real root, in the Bell-value
variable $\lambda$, of the polynomial in Eq.~(18) of that reference:
\begin{align}
F(\lambda,\alpha)={}&\lambda^4+(4-\alpha^2)\lambda^3
+\left(\frac{11}{4}\alpha^4-12\alpha^2-4\right)\lambda^2
\nonumber\\
&+\left(2\alpha^6-\alpha^4-20\alpha^2-32\right)\lambda
\nonumber\\
&+5\alpha^6-21\alpha^4+16\alpha^2-32.
\label{supp:eq:exact-quartic}
\end{align}
Equation~(19) of the same reference gives the common cosine of the two optimal measurement angles,
\begin{equation}
c_*(\alpha)=\frac18\left[3\alpha^2-4+
\sqrt{16+9\alpha^4+8\alpha^2\bigl(2\omega_{\rm Q}(\alpha)-1\bigr)}\right].
\label{supp:eq:exact-optimal-cosine}
\end{equation}
Their behavior near $\alpha=1$ fixes both the vanishing quantum violation and the angles on which to evaluate the certificates.

At $\alpha=1$, the polynomial factors as
\begin{equation}
F(\lambda,1)=\frac14(\lambda-4)
\left(4\lambda^3+28\lambda^2+59\lambda+32\right).
\label{supp:eq:quartic-factor-endpoint}
\end{equation}
There is no real root above $4$, because the cubic factor is positive
there. The quantum value is bounded by $4+2\alpha$ and is at least the
local value $2+2\alpha$. Every limit point as $\alpha\uparrow1$ is
therefore a root of $F(\lambda,1)$ at least $4$, and hence equals $4$.
Thus $\omega_{\rm Q}(1-T)\to4$ as $T\downarrow0$.
Changing the tilt by $T$ changes the Bell value of any strategy by at most
$2T$: the coefficient of the tilt is the sum of two expectations in
$[-1,1]$. The maximum obeys the same bound. Hence the quantum value and
the local value $4-2T$ are both $4+O(T)$.

We next compare the polynomial at those two values. Substituting the
local value gives
\[
F(2+2\alpha,\alpha)=4(\alpha-1)^3(\alpha+2)^4,
\]
whereas $F(\omega_{\rm Q}(\alpha),\alpha)=0$.
At the endpoint, the slope with respect to the Bell-value variable is
\begin{equation}
\left.\frac{\partial F(\lambda,\alpha)}{\partial\lambda}
\right|_{(\lambda,\alpha)=(4,1)}=243.
\label{supp:eq:simple-root}
\end{equation}
To convert the polynomial difference into a difference of Bell values,
apply the mean-value theorem at fixed $\alpha=1-T$. With
$\Delta_{1-T}=\omega_{\rm Q}(1-T)-(4-2T)$, it gives
\[
4T^3(3-T)^4
=\Delta_{1-T}\,\partial_\lambda F(\widetilde\lambda_T,1-T),
\]
for some $\widetilde\lambda_T$ between the local and quantum values.
Both endpoints are $4+O(T)$, so the intermediate value has the same form
and its slope is $243+O(T)$. The numerator has leading coefficient $4\cdot3^4=324$; dividing by the limiting slope gives $324/243=4/3$. Thus
\[
\Delta_{1-T}=\frac43T^3+O(T^4).
\]

The nonzero derivative in Eq.~\eqref{supp:eq:simple-root} also controls
the higher-order remainder. The analytic implicit function theorem
gives a unique real-analytic root $\lambda(T)$ near $4$, with
$\lambda(0)=4$. Since the quantum maximum tends to $4$, it must be this
root:
\begin{equation}
\omega_{\rm Q}(1-T)=\lambda(T)
\label{supp:eq:quantum-analytic-branch}
\end{equation}
for sufficiently small $T\geq0$.
Write
\[
\lambda(T)=4-2T+\frac43T^3+c_4T^4+c_5T^5+O(T^6).
\]
Substitution into $F(\lambda(T),1-T)=0$ gives
\begin{equation}
0=27(9c_4+4)T^4+(243c_5-243c_4-120)T^5+O(T^6).
\label{supp:eq:coefficient-equations}
\end{equation}
The coefficients of $T^4$ and $T^5$ must vanish, giving
$c_4=-4/9$ and $c_5=4/81$. Analyticity bounds the remaining terms:
there are $C,\delta>0$ such that
\begin{align}
\omega_{\rm Q}(1-T)
&=4-2T+\frac43T^3-\frac49T^4+\frac4{81}T^5+R_6(T),
\label{supp:eq:exact-endpoint-expansion}\\
|R_6(T)|&\leq CT^6\qquad(0\leq T<\delta).
\label{supp:eq:R6-bound}
\end{align}
In particular,
\begin{equation}
\lim_{T\downarrow0}
\frac{\omega_{\rm Q}(1-T)-\left(4-2T+\frac43T^3\right)}{T^4}
=-\frac49.
\label{supp:eq:quartic-limit}
\end{equation}
For the detector-loss model of Section~\ref{supp:sec:bell},
$T=(3\eta-2)/\eta=(9/2)(\eta-2/3)+O((\eta-2/3)^2)$. The observed violation is $\eta^2\Delta_\alpha$, so its leading coefficient is $(2/3)^2(4/3)(9/2)^3=54$:

\[
\eta^2[\omega_{\rm Q}(\alpha)-\omega_{\rm L}(\alpha)]
=54(\eta-2/3)^3+O\bigl((\eta-2/3)^4\bigr)
\qquad (\eta\downarrow2/3).
\]

Substituting the quantum-value expansion into
Eq.~\eqref{supp:eq:exact-optimal-cosine} gives
\begin{equation}
c_*(1-T)=1-2T+\frac89T^2+\frac{20}{81}T^3+R_c(T),
\label{supp:eq:optimal-cosine-expansion}
\end{equation}
where $|R_c(T)|\leq C_cT^4$ for sufficiently small $T$.
Since $1-\cos\theta=\theta^2/2+O(\theta^4)$, both optimal angles
satisfy $\theta_*(T)\sim2\sqrt T$. The cited self-testing theorem
identifies the measurement angles of every maximizing strategy, so the
same scale applies to all maximizers in the canonical form.

\section{The NPA hierarchy at level \texorpdfstring{$L$}{L}: moment matrices, duality, and the trace identity}
\label{supp:sec:npa}

\subsection{Words and moment matrices}

The NPA relaxation~\cite{NPA2007,NPA2008,PironioEtAl2010,IoannouRosset2021} keeps the measurement relations while optimizing over the moments of observable products. Write $a_x,b_y$ for the abstract binary observables and $e$ for their identity. Their products form the group
\begin{equation}
G=\langle a_0,a_1,b_0,b_1:
a_x^2=b_y^2=e,\ a_xb_y=b_ya_x\rangle.
\label{supp:eq:group-presentation}
\end{equation}
Moving Alice's generators to the left of Bob's and cancelling adjacent equal generators gives a unique reduced word, alternating within each party. For example, $b_0a_1a_1a_0b_1=a_0b_0b_1$. The group algebra consists of finite complex linear combinations of these words, with multiplication followed by the same reduction. The adjoint reverses each word and conjugates its coefficient. We use $h_\alpha$ from Eq.~\eqref{supp:eq:bell-functional} with $A_x,B_y$ replaced by $a_x,b_y$ in this algebra. Let $\mathcal W_L$ be the set of reduced words of length at most $L$, and set
\begin{equation}
V_L=\Span_{\C}\mathcal W_L.
\label{supp:eq:VL}
\end{equation}

In a quantum strategy, the moment of a word is the expectation of the corresponding product of observables. The relaxation instead assigns a number $\ell(w)$ to each reduced word through length $2L$, without assuming a realizing strategy. Extending the assignment linearly gives a Hermitian functional $\ell$, with $\ell(g^{-1})=\overline{\ell(g)}$. Products of words through length $L$ then define its moment matrix:
\begin{equation}
\Gamma_L(\ell)=\bigl(\ell(u^*v)\bigr)_{u,v\in\mathcal W_L},
\label{supp:eq:moment-matrix}
\end{equation}
For a normalized state $|\psi\rangle$ and an assignment $\pi$ of the generators to measurement observables, the physical moments are
\begin{equation}
\ell(w)=\langle\psi|\pi(w)|\psi\rangle.
\label{supp:eq:quantum-functional}
\end{equation}
For these moments, $\Gamma_L(\ell)$ is the Gram matrix of the vectors $\pi(w)|\psi\rangle$, hence positive semidefinite. Maximizing the candidate Bell value $\ell(h_\alpha)$ under this necessary condition gives the level-$L$ NPA bound:
\begin{equation}
\omega_{\rm NPA}^{(L)}(\alpha)=
\sup\left\{\ell(h_\alpha):\ell(e)=1,\ \Gamma_L(\ell)\succeq0\right\}.
\label{supp:eq:primal-NPA}
\end{equation}
Every quantum strategy is feasible, so
\begin{equation}
\omega_{\rm NPA}^{(L)}(\alpha)\geq\omega_{\rm Q}(\alpha).
\label{supp:eq:outer-bound}
\end{equation}
At $L=1$ this observable moment program is an algebraic upper bound; its moments need not reconstruct nonnegative probabilities. The extraction of valid correlations in Section~\ref{supp:sec:transfer} uses $d\geq2$.

If $p=\sum_{w\in\mathcal W_L}c_ww$ has coefficient vector $\bm c=(c_w)_w$, then $\ell(p^*p)=\bm c^*\Gamma_L(\ell)\bm c$. Thus moment positivity is equivalent to nonnegativity on every square $p^*p$ with $p\in V_L$, and hence on their sums. This moment/SOS viewpoint is standard in Bell optimization and self-testing~\cite{BampsPironio2015,KaniewskiCEQIP2018}. These sums form the cone
\begin{equation}
\Sigma_L=
\left\{\sum_{j=1}^{r}p_j^*p_j:
p_j\in V_L,\ r<\infty\right\}.
\label{supp:eq:sigma-L}
\end{equation}
Writing each $p_j$ as a row of coefficients times the word vector $\bm w$ gives $Q=C^*C$, where $C$ collects those rows. Hence every sum has the form
\begin{equation}
\sum_jp_j^*p_j=\bm w^*Q\bm w,
\qquad Q\succeq0.
\label{supp:eq:gram-representation}
\end{equation}
If $\lambda e-h_\alpha\in\Sigma_L$, every feasible $\ell$ satisfies $\lambda-\ell(h_\alpha)=\sum_j\ell(p_j^*p_j)\geq0$. To compare the certificate with the NPA error, we need this identity at the optimal bound itself. The strict-feasibility and boundedness conditions below give that attainment by SDP duality.

\subsection{Applying SDP duality}

\begin{proposition}
\label{supp:prop:duality}
For every $L\geq1$ and $0\leq\alpha\leq1$,
\begin{equation}
\omega_{\rm NPA}^{(L)}(\alpha)=
\min\left\{\lambda\in\R:\lambda e-h_\alpha\in\Sigma_L\right\}.
\label{supp:eq:dual-NPA}
\end{equation}
The primal maximum and the dual minimum are attained, and the Gram matrix in Eq.~\eqref{supp:eq:gram-representation} may be chosen real and symmetric.
\end{proposition}

\begin{proof}
The canonical trace extracts the identity coefficient of a word polynomial:
\begin{equation}
\tau\!\left(\sum_{g\in G}c_gg\right)=c_e.
\label{supp:eq:canonical-trace}
\end{equation}
Since $u^*v=e$ exactly when $u=v$, we have $\tau(u^*v)=\delta_{u,v}$ and hence
\begin{equation}
\Gamma_L(\tau)=I\succ0.
\label{supp:eq:strict-feasibility}
\end{equation}
Thus the positivity constraint has a strictly feasible point~\cite[Appendix~C]{Tavakoli2024}. To show that the optimum is finite and attained, consider the $2\times2$ principal minors of any feasible moment matrix. They give
\begin{equation}
|\ell(u^*v)|^2\leq\ell(u^*u)\,\ell(v^*v)=1,
\label{supp:eq:moment-CS}
\end{equation}
using $u^*u=v^*v=e$. Every reduced word $g$ of length at most $2L$ can be written as $g=u^*v$ with $u,v\in\mathcal W_L$: split a reduced representative as $g=st$ with $|s|,|t|\leq L$ and take $u=s^{-1}$, $v=t$. Equation~\eqref{supp:eq:moment-CS} therefore bounds every coordinate of $\ell$ on $\Span\mathcal W_{2L}$ by one. The feasible set is closed and bounded in a finite-dimensional space, hence compact, and the supremum in Eq.~\eqref{supp:eq:primal-NPA} is a maximum.

Strict feasibility and the finite primal optimum give zero duality gap and dual attainment by Slater's theorem~\cite{BoydVandenberghe2004}; this uses the standard SDP duality reviewed in Ref.~\cite{VandenbergheBoyd1996}. The dual variable is a positive semidefinite matrix $Q$ satisfying
\begin{equation}
\lambda e-h_\alpha=\bm w^*Q\bm w,
\qquad Q\succeq0,
\label{supp:eq:dual-gram}
\end{equation}
which is the statement $\lambda e-h_\alpha\in\Sigma_L$. Finally, the multiplication table of the reduced-word basis and the coefficients of $h_\alpha$ are real; if $Q$ is a complex Hermitian Gram matrix, then $(Q+\overline Q)/2$ is real, symmetric, positive semidefinite, and represents the same element.
\end{proof}

\subsection{The trace identity}

The coefficients in an optimal certificate can vary with the tilt. The same canonical trace bounds their square-root matrix, because it extracts the identity coefficient on both sides of the certificate.

\begin{lemma}[Trace identity]
\label{supp:lem:trace-identity}
Every Gram representation
\begin{equation}
\lambda e-h_\alpha=\bm w^*Q\bm w,
\qquad Q\succeq0,
\label{supp:eq:generic-dual-gram}
\end{equation}
satisfies
\begin{equation}
\Tr Q=\lambda.
\label{supp:eq:trace-Q}
\end{equation}
In particular, $\|Q\|\leq\lambda$ and $\|Q^{1/2}\|\leq\sqrt\lambda$.
\end{lemma}

\begin{proof}
Apply the canonical trace to both sides of Eq.~\eqref{supp:eq:generic-dual-gram}. Because $\mathcal W_L$ contains one representative of each reduced word, $\tau(w_i^*w_j)=\delta_{i,j}$,
\begin{equation}
\tau(\bm w^*Q\bm w)
=\sum_{i,j}Q_{ij}\,\tau(w_i^*w_j)
=\sum_iQ_{ii}
=\Tr Q.
\label{supp:eq:trace-gram-calculation}
\end{equation}
Since $\tau(h_\alpha)=0$, the left side equals $\lambda$. The norm bounds follow from $Q\succeq0$ and $\|Q\|\leq\Tr Q$.
\end{proof}

\section{The Motzkin polynomial from the Bell operator}
\label{supp:sec:motzkin}

Following the angle scale in Section~\ref{supp:sec:endpoint}, write
\begin{equation}
T=t^2,
\label{supp:eq:T-t}
\end{equation}
and allow both measurement angles to vary independently. With real scaled angular coordinates $x$ and $y$ (denoted $u$ and $v$ in the Letter), take
\begin{align}
A_0&=Z,&A_1&=\cos(tx)Z+\sin(tx)X,\nonumber\\
B_0&=Z,&B_1&=\cos(ty)Z+\sin(ty)X.
\label{supp:eq:two-parameter-observables}
\end{align}
Let $H(t;x,y)$ denote the matrix of $h_{1-t^2}$ evaluated on these observables, and write $H(t)$ when the angles are fixed. This is the signed-sine parametrization of the Bell matrix in Section~\ref{supp:sec:jordan}, with $T=t^2$. The maximizing angles correspond asymptotically to $(x,y)=(2,2)$, but the certificate is an operator identity and must hold for every real $(x,y)$. Keeping both coordinates free will give a polynomial identity in two variables.

The violation above the local bound is of order $t^6$. To obtain a gap between the quantum bound and the attained Bell value of the same order, we seek a vector whose expectation of the positive operator $\omega_{\rm Q}(1-t^2)I-H(t)$ is of order $t^6$. We first keep the vector unnormalized so that its entries can be polynomials; its normalization will give the physical state below.

At $t=0$, $H(0)=\operatorname{diag}(4,-2,-2,0)$, with top eigenvector $|00\rangle$ and gaps $6,6,4$ to $|01\rangle,|10\rangle,|11\rangle$. Substituting the angles into Eq.~\eqref{supp:eq:block-matrix} gives the couplings away from this vector:
\[
\begin{aligned}
\langle01|H(t)|00\rangle&=[1-\cos(tx)]\sin(ty)
=\tfrac12 x^2y\,t^3+O_K(t^5),\\
\langle10|H(t)|00\rangle&=[1-\cos(ty)]\sin(tx)
=\tfrac12 xy^2\,t^3+O_K(t^5),\\
\langle11|H(t)|00\rangle&=-\sin(tx)\sin(ty)
=-xy\,t^2+O_K(t^4).
\end{aligned}
\]
A real correction $z$ along $|11\rangle$ contributes $4z^2+2xy\,t^2z$ to the leading gap, so its minimizing order is $t^2$. The other two couplings have order $t^3$ and diagonal gaps $6$, giving corrections of order $t^3$. We therefore vary three real coefficients $A,B,C$ at fixed $(x,y)$ in the vector
\begin{equation}
\zeta_t=(1,At^3,Bt^3,Ct^2)^T.
\label{supp:eq:generic-scaled-vector}
\end{equation}
At the required order, we can replace the quantum value by
\begin{equation}
\Lambda_t=4-2t^2+\frac43t^6.
\label{supp:eq:Lambda-t}
\end{equation}
Equation~\eqref{supp:eq:exact-endpoint-expansion} bounds its difference from the exact quantum value by $O(t^8)$. To expand the expectation, use Eq.~\eqref{supp:eq:excess-quadratic-form} with $u=1-\cos(tx)$, $v=1-\cos(ty)$, $r=\sin(tx)$, $s=\sin(ty)$, and $(b_1,b_2,b_3)=(At^3,Bt^3,Ct^2)$. Since $\Lambda_t-\lambda_0(t^2)=4t^6/3$, the expectation is $4t^6\|\zeta_t\|^2/3-\mathcal E_{t^2}(\zeta_t)$. At order $t^4$, the $|00\rangle$ diagonal, its coupling to $|11\rangle$, and the squared $|11\rangle$ correction contribute $x^2y^2/4$, $2Cxy$, and $4C^2$, respectively. Their sum is a square:
\begin{equation}
\langle\zeta_t,(\Lambda_tI-H(t))\zeta_t\rangle
=\frac{(4C+xy)^2}{4}t^4+O(t^6).
\label{supp:eq:generic-gap-expansion}
\end{equation}
A finite limit after division by $t^6$ requires $C=-xy/4$: the $|11\rangle$ correction cancels the entire order-$t^4$ contribution. After this choice, the coefficient of $t^6$ obtained from Eq.~\eqref{supp:eq:excess-quadratic-form} is
\[
6A^2-Ax^2y+6B^2-Bxy^2
+\frac43+\frac{x^4y^2+x^2y^4}{16}-\frac{x^2y^2}{4}.
\]
Completing the square in $A$ subtracts $x^4y^2/24$ from the terms independent of $A$, and completing the square in $B$ subtracts $x^2y^4/24$. Since $1/16-1/24=1/48$, the expectation becomes
\begin{align}
\langle\zeta_t,(\Lambda_tI-H(t))\zeta_t\rangle
={}&t^6\biggl[
6\left(A-\frac{x^2y}{12}\right)^2
+6\left(B-\frac{xy^2}{12}\right)^2
\nonumber\\
&\hspace{12mm}+\frac43+\frac{x^4y^2+x^2y^4}{48}-\frac{x^2y^2}{4}
\biggr]+O(t^8).
\label{supp:eq:gap-completed-squares}
\end{align}
Choosing $A=x^2y/12$ and $B=xy^2/12$ cancels the two squares. Together with $C=-xy/4$, this gives the polynomial vector
\begin{equation}
\xi_t(x,y)=
\begin{pmatrix}
1&t^3x^2y/12&t^3xy^2/12&-t^2xy/4
\end{pmatrix}^{T}.
\label{supp:eq:scaled-test-vector}
\end{equation}
This vector also has a spectral interpretation. On a compact angular set $K$, $H(t)$ stays uniformly close to $H(0)$ for small $t$, so the simple top eigenvalue and its eigenvector vary analytically. Normalize the first component of that eigenvector to one. The leading equations for the other components are $6A=x^2y/2$, $6B=xy^2/2$, and $4C=-xy$, using the couplings and gaps above. Thus $\xi_t$ is its Taylor polynomial through order three, with an $O_K(t^4)$ error. The required gap estimate follows directly from the completed squares, as the next proposition shows.

\begin{proposition}
\label{supp:prop:motzkin-limit}
Locally uniformly in $(x,y)\in\R^2$,
\begin{equation}
\lim_{t\downarrow0}t^{-6}
\langle\xi_t,[\omega_{\rm Q}(1-t^2)I-H(t;x,y)]\xi_t\rangle
=\frac43M\!\left(\frac{x}{2},\frac{y}{2}\right),
\label{supp:eq:motzkin-limit}
\end{equation}
where
\begin{equation}
M(X,Y)=X^4Y^2+X^2Y^4-3X^2Y^2+1.
\label{supp:eq:Motzkin-polynomial}
\end{equation}
\end{proposition}

\begin{proof}
Fix a compact $K\subset\R^2$. The sine and cosine Taylor remainders are uniform on $K$. To control the first omitted order, note that conjugation by $U=Z\otimes Z$ reverses the sign of each $X$ factor. Hence
\begin{equation}
H(-t;x,y)=UH(t;x,y)U,
\qquad
\xi_{-t}(x,y)=U\xi_t(x,y),
\label{supp:eq:t-evenness}
\end{equation}
so the expectation is even in $t$. The Taylor coefficients calculated above therefore leave an $O_K(t^8)$ remainder. Substituting the entries of $\xi_t$ from Eq.~\eqref{supp:eq:scaled-test-vector} into Eq.~\eqref{supp:eq:gap-completed-squares}, we obtain
\begin{align}
\langle\xi_t,(\Lambda_tI-H(t))\xi_t\rangle
=t^6\left[
\frac43+\frac{x^4y^2}{48}+\frac{x^2y^4}{48}-\frac{x^2y^2}{4}
\right]+O_K(t^8).
\label{supp:eq:motzkin-leading-term}
\end{align}
The bracket is $\frac43M(x/2,y/2)$. Restoring the exact quantum value uses
\begin{equation}
\omega_{\rm Q}(1-t^2)-\Lambda_t
=-\frac49t^8+\frac4{81}t^{10}+R_6(t^2),
\label{supp:eq:omega-Lambda}
\end{equation}
where $|R_6(t^2)|\leq Ct^{12}$ for sufficiently small $t$. The exact norm is
\begin{equation}
\|\xi_t(x,y)\|^2
=1+\frac{x^2y^2}{16}t^4+\frac{x^4y^2+x^2y^4}{144}t^6
=1+O_K(t^4).
\label{supp:eq:test-vector-norm}
\end{equation}
Multiplying Eq.~\eqref{supp:eq:omega-Lambda} by this norm changes the expectation by $O_K(t^8)$. After division by $t^6$ the change is $O_K(t^2)$, which proves Eq.~\eqref{supp:eq:motzkin-limit}.
\end{proof}

The first component of $\xi_t$ is one, so it is nonzero for every $(t,x,y)$. Normalizing it gives a two-qubit state and its Bell value on the observables in Eq.~\eqref{supp:eq:two-parameter-observables}:
\begin{equation}
\psi_t(x,y)=\frac{\xi_t(x,y)}{\|\xi_t(x,y)\|},
\qquad
b_t(x,y)=\langle\psi_t(x,y),H(t;x,y)\psi_t(x,y)\rangle.
\label{supp:eq:normalized-test-strategy}
\end{equation}
Write $\Delta_\alpha=\omega_{\rm Q}(\alpha)-\omega_{\rm L}(\alpha)$ for the maximal quantum violation above the local bound. Equation~\eqref{supp:eq:exact-endpoint-expansion} gives $\Delta_{1-t^2}=4t^6/3+O(t^8)$, while Eq.~\eqref{supp:eq:test-vector-norm} gives $\|\xi_t\|^2=1+O_K(t^4)$. Therefore the gap between the exact quantum bound and this attained Bell value satisfies
\begin{align}
\frac{\omega_{\rm Q}(1-t^2)-b_t(x,y)}{\Delta_{1-t^2}}
&=\frac{\langle\xi_t,[\omega_{\rm Q}(1-t^2)I-H(t;x,y)]\xi_t\rangle}
{\|\xi_t\|^2\Delta_{1-t^2}}
\nonumber\\
&\xrightarrow[t\downarrow0]{}M\!\left(\frac{x}{2},\frac{y}{2}\right)
\label{supp:eq:normalized-bell-gap}
\end{align}
locally uniformly in $(x,y)$. If $x=0$ or $y=0$, then $\xi_t=|00\rangle$ and $H(t)\xi_t=\omega_{\rm L}(1-t^2)\xi_t$. Thus the numerator equals $\Delta_{1-t^2}$, the ratio is exactly $1$, and $M$ is also $1$ on either axis. At $(x,y)=(2,2)$, the ratio tends to $M(1,1)=0$: these strategies attain the quantum bound up to an error smaller than the violation above the local bound. A certificate evaluated on these strategies adds the NPA error $\omega_{\rm NPA}^{(L)}(\alpha)-\omega_{\rm Q}(\alpha)$ to this strategy gap. We will compare that error with the same vanishing quantum violation.

The limit $M$ is the Motzkin polynomial~\cite{Motzkin1967}. It is nonnegative by the arithmetic-geometric mean inequality applied to $X^4Y^2$, $X^2Y^4$, and $1$,
\begin{equation}
X^4Y^2+X^2Y^4+1\geq3X^2Y^2.
\label{supp:eq:motzkin-nonnegative}
\end{equation}

The equal-angle restriction is itself a sum of squares:
\[
M(X,X)=(X^2-1)^2+2[X(X^2-1)]^2.
\]
Adding a nonnegative constant preserves this property, so the contradiction below uses the polynomial identity on the full two-parameter family.

\section{The Motzkin polynomial plus a nonnegative constant is not a sum of squares}
\label{supp:sec:not-sos}

A bounded NPA error on the $t^6$ scale contributes a nonnegative constant to the limit because $\|\xi_t\|^2\to1$. We must therefore exclude sums of squares for every $M+c$, $c\geq0$~\cite{Hilbert1888,Motzkin1967}.

The elementary argument used in the Letter already proves the obstruction. Suppose $M+c=\sum_jq_j^2$. Each $q_j$ has degree at most three: the highest-degree homogeneous terms of their squares cannot cancel unless all corresponding highest-degree terms vanish. On either axis, $M+c=1+c$. Since every square is nonnegative, each restricted $q_j$ is bounded on the whole real line and is therefore constant. Every $q_j$ therefore has the form $a_j+b_jXY+c_jX^2Y+d_jXY^2$. The coefficient of $X^2Y^2$ in their squared sum is $\sum_jb_j^2\geq0$, whereas it is $-3$ in $M+c$.

The same monomial restriction follows from the Newton-polytope theorem~\cite[Theorem~1]{Reznick1978}, giving an alternative proof. For $p(X,Y)=\sum_{i,j}p_{ij}X^iY^j$, its Newton polytope is $\Newt(p)=\operatorname{conv}\{(i,j):p_{ij}\neq0\}$, the convex hull of the exponents with nonzero coefficients.

\begin{lemma}[Half-Newton-polytope property~\cite{Reznick1978}]
\label{supp:lem:half-newton}
If $p=\sum_kq_k^2$ with real polynomials $q_k$, then
\begin{equation}
\Newt(q_k)\subseteq\frac12\Newt(p)
\label{supp:eq:half-newton}
\end{equation}
for every $k$.
\end{lemma}

The cited theorem is stated for homogeneous forms. Homogenizing the polynomials to a common even degree and then setting the extra variable to one gives Eq.~\eqref{supp:eq:half-newton}.

\begin{lemma}
\label{supp:lem:motzkin-not-sos}
For every $c\geq0$, the polynomial $M(X,Y)+c$ is not a sum of squares of real polynomials.
\end{lemma}

\begin{proof}
The constant coefficient of $M+c$ is $1+c>0$, so its exponent set is $\{(0,0),(4,2),(2,4),(2,2)\}$ and
\begin{equation}
\Newt(M+c)=\operatorname{conv}\{(0,0),(4,2),(2,4)\},
\label{supp:eq:motzkin-newton}
\end{equation}
the point $(2,2)$ being interior to the triangle. Half of this triangle has vertices $(0,0)$, $(2,1)$, and $(1,2)$ and exactly four integer points:

\begin{equation}
(0,0),\quad(1,1),\quad(2,1),\quad(1,2).
\label{supp:eq:half-triangle-points}
\end{equation}
By Lemma~\ref{supp:lem:half-newton}, every polynomial in a hypothetical representation $M+c=\sum_kq_k^2$ has the form
\begin{equation}
q_k=a_k+b_kXY+c_kX^2Y+d_kXY^2.
\label{supp:eq:qk-form}
\end{equation}
The only pair of allowed exponents summing to $(2,2)$ is $(1,1)+(1,1)$. Hence the coefficient of $X^2Y^2$ in the sum of squares is
\begin{equation}
\sum_kb_k^2\geq0,
\label{supp:eq:positive-middle-coefficient}
\end{equation}
whereas the coefficient of $X^2Y^2$ in $M+c$ is $-3$. This contradiction proves the lemma.
\end{proof}

\section{Fixed NPA levels yield polynomial sum-of-squares limits}
\label{supp:sec:compactness}

A certificate evaluated on the two-qubit strategies gives a sum of squared norms. To apply the obstruction in Section~\ref{supp:sec:not-sos}, we need its limit to be a sum of polynomial squares. At fixed level, the factors contain only finitely many words, and their Taylor coefficients are polynomials in the angles. The following argument controls the error of replacing the factors by these polynomials before taking the limit.

Fix $L$ and enumerate $\mathcal W_L$ as $w_1,\ldots,w_N$. Let $\pi_{t,x,y}$ evaluate each word on the observables in Eq.~\eqref{supp:eq:two-parameter-observables}. Applying all words to the polynomial vector $\xi_t$ gives
\begin{equation}
z_t(x,y)=
\begin{pmatrix}
\pi_{t,x,y}(w_1)\xi_t(x,y)\\
\vdots\\
\pi_{t,x,y}(w_N)\xi_t(x,y)
\end{pmatrix}
\in\C^{4N},
\label{supp:eq:z-vector}
\end{equation}
If $C_n=Q_n^{1/2}$, the square factor $f_{n,j}=\sum_k(C_n)_{jk}w_k$ acts on $\xi_t$ by combining the corresponding four-component blocks of $z_t$. Stacking the results gives $g_n(t,x,y)=(C_n\otimes I_4)z_t(x,y)$. We keep $\xi_t$ unnormalized to retain polynomial entries; its norm tends locally uniformly to one, so normalization leaves any finite rescaled squared-norm limit unchanged.

\begin{lemma}[Fixed-level sum-of-squares limit]
\label{supp:lem:fixed-level-limit}
Let $t_n\downarrow0$ and let the $N\times N$ matrices $Q_n\succeq0$ have uniformly bounded trace. Suppose that for an integer $m\geq0$,
\begin{equation}
t_n^{-2m}\left\|(Q_n^{1/2}\otimes I_4)z_{t_n}(x,y)\right\|^2
\longrightarrow P(x,y)
\label{supp:eq:compactness-assumption}
\end{equation}
locally uniformly, where $P$ is a real polynomial. Then $P$ is a sum of squares of real polynomials.
\end{lemma}

\begin{proof}
Each observable in Eq.~\eqref{supp:eq:two-parameter-observables} is analytic in $t$. A sine or cosine coefficient of order $k$ is a polynomial of degree at most $k$ in its angular coordinate, and the entries of $\xi_t$ have the same relation between their order in $t$ and their degree in $x,y$. In a product contributing to order $r$, these orders add to $r$. Thus, for the fixed word list,
\begin{equation}
z_t(x,y)=\sum_{r=0}^{m}t^rz_r(x,y)+R_{m+1}(t,x,y),
\label{supp:eq:z-Taylor}
\end{equation}
where each coordinate of $z_r$ is a polynomial of total degree at most $r$. The word list is finite, so the analytic Taylor remainders can be bounded simultaneously: for every compact $K\subset\R^2$,
\begin{equation}
\sup_{(x,y)\in K}\|R_{m+1}(t,x,y)\|=O_K(t^{m+1}).
\label{supp:eq:remainder-bound}
\end{equation}

Take $C_n=Q_n^{1/2}$, keeping the vector dimension fixed even when the rank changes, and let $C_{\rm tr}=\sup_n\Tr Q_n<\infty$. Positivity gives
\begin{equation}
\|C_n\|\leq\sqrt{C_{\rm tr}}.
\label{supp:eq:Cn-bound}
\end{equation}
The vector $v_n=t_n^{-m}(C_n\otimes I_4)z_{t_n}$ has the convergent squared norm in Eq.~\eqref{supp:eq:compactness-assumption}. Retaining the Taylor terms through order $m$ gives its polynomial approximation
\begin{equation}
p_n(x,y)=t_n^{-m}(C_n\otimes I_4)
\sum_{r=0}^{m}t_n^rz_r(x,y).
\label{supp:eq:pn-definition}
\end{equation}
Equations~\eqref{supp:eq:remainder-bound} and \eqref{supp:eq:Cn-bound} give, for all sufficiently large $n$,
\begin{equation}
\sup_K\left\|p_n-t_n^{-m}(C_n\otimes I_4)z_{t_n}\right\|
\leq\sqrt{C_{\rm tr}}\,C_Kt_n.
\label{supp:eq:pn-remainder}
\end{equation}
Thus $p_n-v_n$ tends uniformly to zero on every compact set. The coefficients of $p_n$ may still contain negative powers of $t_n$; their bound will instead follow from the bounded values of $v_n$. The two bounds now have separate roles: the trace bound controls the Taylor remainder, while the assumed squared-norm limit controls the values of the rescaled polynomial factors.

Take $K_0=[-1,1]^2$. The convergence in Eq.~\eqref{supp:eq:compactness-assumption} bounds $v_n$ on $K_0$, and Eq.~\eqref{supp:eq:pn-remainder} then bounds $p_n$ there. Every $p_n$ has $4N$ coordinates of total degree at most $m$. We can recover their coefficients from finitely many values: choose $m+1$ fixed distinct nodes $s_0,\ldots,s_m$ in $[-1,1]$, with Lagrange polynomials
\[
L_i(x)=\prod_{k\neq i}\frac{x-s_k}{s_i-s_k}.
\]
Each component $p$ of $p_n$ has degree at most $m$ in each variable, so
\[
p(x,y)=\sum_{i,j=0}^m p(s_i,s_j)L_i(x)L_j(y).
\]
The coefficients of the fixed polynomials $L_iL_j$ do not depend on $n$. The bounded grid values therefore bound every monomial coefficient of $p_n$. Writing $\|p_n\|_{\rm coeff}$ for the Euclidean norm of this finite list of coefficients, we obtain
\begin{equation}
\|p_n\|_{\rm coeff}
\leq C_{m,N}\sup_{K_0}\|p_n\|\leq C.
\label{supp:eq:rescaled-coefficient-bound}
\end{equation}
This is simply norm equivalence on the fixed polynomial space: $\sup_{K_0}\|p\|$ is a norm because a polynomial that vanishes on the interior of $K_0$ vanishes identically. For the application in the Letter, $m=3$, so four nodes in each variable suffice.

There are finitely many bounded coefficients, so one subsequence makes them all converge, defining a polynomial vector $q(x,y)$. On any compact set the finitely many monomials are bounded, so convergence of their coefficients gives uniform convergence of $p_n$ to $q$ there.

To identify the squared norm of $q$, use the approximation $p_n-v_n\to0$. Both vectors are bounded on $K_0$, so
\begin{equation}
\bigl|\|p_n\|^2-\|v_n\|^2\bigr|
\leq(\|p_n\|+\|v_n\|)\,\|p_n-v_n\|
\longrightarrow0
\label{supp:eq:squared-norm-comparison}
\end{equation}
uniformly on $K_0$. Combining this estimate with Eq.~\eqref{supp:eq:compactness-assumption} and the coefficientwise convergence of $p_n$ gives
\begin{equation}
P(x,y)=\|q(x,y)\|^2
\label{supp:eq:P-norm-square}
\end{equation}
on $K_0$. The difference of the two sides is a polynomial vanishing on the interior of $K_0$, so the identity holds on all of $\R^2$. Finally, splitting each complex coordinate of $q$ into its real and imaginary parts writes its squared modulus as two squares of real polynomials.
\end{proof}

\section{No finite level is exact: proof of the divergence}
\label{supp:sec:main-proof}

If the error of a fixed NPA level stayed bounded on the cubic scale, its optimal certificates would satisfy Lemma~\ref{supp:lem:fixed-level-limit}. Their limiting polynomial would then contradict Section~\ref{supp:sec:not-sos}.

\begin{theorem}
\label{supp:thm:fixed-level}
For every finite $L\geq1$,
\begin{equation}
\lim_{\alpha\uparrow1}
\frac{\omega_{\rm NPA}^{(L)}(\alpha)-\omega_{\rm Q}(\alpha)}{(1-\alpha)^3}=+\infty.
\label{supp:eq:main-divergence}
\end{equation}
\end{theorem}

\begin{proof}
Fix $L$ throughout the proof. By Eq.~\eqref{supp:eq:outer-bound}, the error is nonnegative. Failure of the claimed divergence would give a sequence approaching the critical tilt along which the error divided by $(1-\alpha)^3$ is bounded. Writing $1-\alpha=t_n^2$, there would therefore exist $t_n\downarrow0$ and bounded $r_n\geq0$ with
\begin{equation}
\omega_{\rm NPA}^{(L)}(1-t_n^2)=\omega_{\rm Q}(1-t_n^2)+r_nt_n^6.
\label{supp:eq:bounded-excess}
\end{equation}
After passing to a further subsequence, assume $r_n\to r\geq0$.

Dual attainment in Proposition~\ref{supp:prop:duality} realizes these exact level-$L$ values with Gram certificates on the fixed word list:
\begin{equation}
\omega_{\rm NPA}^{(L)}(1-t_n^2)\,e-h_{1-t_n^2}=\bm w^*Q_n\bm w.
\label{supp:eq:certificate-sequence}
\end{equation}
Lemma~\ref{supp:lem:trace-identity} gives
\begin{equation}
\Tr Q_n=\omega_{\rm NPA}^{(L)}(1-t_n^2).
\label{supp:eq:trace-sequence}
\end{equation}
The Cauchy--Schwarz inequality for moments, Eq.~\eqref{supp:eq:moment-CS}, bounds each of the four correlators and each marginal by one, so $\omega_{\rm NPA}^{(L)}(\alpha)\leq6$ for all $0\leq\alpha\leq1$, and the traces in Eq.~\eqref{supp:eq:trace-sequence} are uniformly bounded.

Evaluating Eq.~\eqref{supp:eq:certificate-sequence} on the observables in Eq.~\eqref{supp:eq:two-parameter-observables} and the vector $\xi_{t_n}(x,y)$ turns the Gram form into a squared norm:
\begin{align}
&\langle\xi_{t_n},
[\omega_{\rm NPA}^{(L)}(1-t_n^2)I-H(t_n;x,y)]\xi_{t_n}\rangle
\nonumber\\
&\hspace{22mm}
=\left\|(Q_n^{1/2}\otimes I_4)z_{t_n}(x,y)\right\|^2,
\label{supp:eq:evaluated-certificate}
\end{align}
where $z_t$ is defined in Eq.~\eqref{supp:eq:z-vector}. By Eq.~\eqref{supp:eq:bounded-excess}, the expectation on the left is the strategy gap $\langle\xi_{t_n},[\omega_{\rm Q}(1-t_n^2)I-H(t_n)]\xi_{t_n}\rangle$ plus $r_nt_n^6\|\xi_{t_n}\|^2$. Dividing by $t_n^6$, the first term tends to $\frac43M(x/2,y/2)$ by Eq.~\eqref{supp:eq:motzkin-limit}, and the second tends to $r$ because $\|\xi_{t_n}\|^2\to1$. Consequently,
\begin{equation}
t_n^{-6}
\left\|(Q_n^{1/2}\otimes I_4)z_{t_n}(x,y)\right\|^2
\longrightarrow
\frac43M\!\left(\frac{x}{2},\frac{y}{2}\right)+r
\label{supp:eq:certificate-limit}
\end{equation}
locally uniformly. The matrices have bounded trace, the word list is fixed, and the right-hand side is a real polynomial. All hypotheses of Lemma~\ref{supp:lem:fixed-level-limit} therefore hold with $m=3$, so this limiting polynomial is a sum of squares of real polynomials. If its real square factors are $q_j(x,y)$, substituting $x=2X$, $y=2Y$ and multiplying by $3/4$ gives
\begin{equation}
M(X,Y)+\frac{3r}{4}
=\sum_j\left[\frac{\sqrt3}{2}q_j(2X,2Y)\right]^2,
\label{supp:eq:forbidden-SOS}
\end{equation}
contradicting Lemma~\ref{supp:lem:motzkin-not-sos}. Hence Eq.~\eqref{supp:eq:main-divergence} holds.
\end{proof}

The local value in Eq.~\eqref{supp:eq:local-value} and the expansion in Eq.~\eqref{supp:eq:exact-endpoint-expansion} give
\begin{equation}
\lim_{\alpha\uparrow1}
\frac{\omega_{\rm Q}(\alpha)-\omega_{\rm L}(\alpha)}{(1-\alpha)^3}
=\frac43.
\label{supp:eq:quantum-advantage}
\end{equation}
Combining this limit with Theorem~\ref{supp:thm:fixed-level} gives
\begin{equation}
\lim_{\alpha\uparrow1}
\frac{\omega_{\rm NPA}^{(L)}(\alpha)-\omega_{\rm Q}(\alpha)}
{\omega_{\rm Q}(\alpha)-\omega_{\rm L}(\alpha)}
=+\infty.
\label{supp:eq:relative-divergence}
\end{equation}
For each fixed $L$, the moment constraints keep the absolute NPA error bounded, while the quantum violation vanishes cubically. Equation~\eqref{supp:eq:relative-divergence} shows that the error exceeds every constant multiple of that violation.

\section{From Bell values to correlation sets}
\label{supp:sec:transfer}

A Bell value above $\omega_{\rm Q}$ certifies a nonquantum correlation only when the moments define valid probabilities. We first show that level $d\geq2$ guarantees this. We then choose $d$ large enough that the same moments satisfy the constraints of any prescribed finite projector-word list.

\begin{definition}
\label{supp:def:finite-word-relaxation}
Fix a finite list $\mathcal S=(S_1,\ldots,S_m)$ of words in the projectors $E_{a|x}$ and $F_{b|y}$. We say that a valid no-signaling correlation $p$ belongs to $\mathcal N_{\mathcal S}$ if there exists a Hermitian linear functional $\ell$, defined on all words required to form the moment matrix, impose the identities below, and recover the observed probabilities, such that
\begin{equation}
\ell(I)=1,
\qquad
\Gamma_{\mathcal S}=\bigl(\ell(S_i^*S_j)\bigr)_{i,j=1}^{m}\succeq0,
\label{supp:eq:selected-moment-matrix}
\end{equation}
and $\ell$ is required to obey the identities generated by
\begin{align}
E_{a|x}E_{a'|x}&=\delta_{a,a'}E_{a|x},
&\sum_aE_{a|x}&=I,
\label{supp:eq:alice-projector-identities}\\
F_{b|y}F_{b'|y}&=\delta_{b,b'}F_{b|y},
&\sum_bF_{b|y}&=I,
\label{supp:eq:bob-projector-identities}\\
[E_{a|x},F_{b|y}]&=0,
&\ell(W^*)&=\overline{\ell(W)}.
\label{supp:eq:cross-and-adjoint}
\end{align}
The observed probabilities are linked to the moments by
\begin{equation}
p(a,b|x,y)=\ell(E_{a|x}F_{b|y}).
\label{supp:eq:behavior-link}
\end{equation}
\end{definition}

Every standard finite NPA level is included in this definition. If the chosen list does not constrain some of the moments in Eq.~\eqref{supp:eq:behavior-link}, those moments remain free apart from the stated identities and the requirement that $p$ is a valid correlation. Additional localizing matrices, state-optimality conditions, operator-optimality conditions, or problem-specific inequalities define different relaxations.

\begin{lemma}[Valid correlations at level $d\geq2$]
\label{supp:lem:valid-behavior}
Let $d\geq2$ and let $\ell$ be feasible for the level-$d$ moment matrix $\Gamma_d(\ell)$. Define the abstract projectors
\begin{equation}
E_{a|x}=\frac{e+aa_x}{2},
\qquad
F_{b|y}=\frac{e+bb_y}{2},
\label{supp:eq:abstract-projectors}
\end{equation}
and
\begin{equation}
p(a,b|x,y)=\ell(E_{a|x}F_{b|y}).
\label{supp:eq:extracted-probability}
\end{equation}
Then $p$ is nonnegative, normalized, and no-signaling.
\end{lemma}

\begin{proof}
The relations $a_x^2=b_y^2=e$ and $[a_x,b_y]=0$ make $P_{ab|xy}=E_{a|x}F_{b|y}$ a self-adjoint projector. Expanding Eq.~\eqref{supp:eq:abstract-projectors} writes it as $(e+aa_x+bb_y+ab\,a_xb_y)/4$, using only words of length at most two. Hence $P_{ab|xy}\in V_d$ for $d\geq2$. Moment positivity therefore gives
\begin{equation}
p(a,b|x,y)=\ell(P_{ab|xy})
=\ell(P_{ab|xy}^*P_{ab|xy})\geq0.
\label{supp:eq:probability-positive}
\end{equation}
Using $\sum_aE_{a|x}=\sum_bF_{b|y}=e$, we obtain normalization and Alice's marginal:
\begin{align}
\sum_{a,b}p(a,b|x,y)&=\ell(e)=1,
\label{supp:eq:probability-normalization}\\
\sum_b p(a,b|x,y)&=\ell(E_{a|x}).
\label{supp:eq:no-signaling-A}
\end{align}
The marginal is independent of $y$; the analogous identity for Bob proves no-signaling.
\end{proof}

\begin{lemma}[Containment in a relaxation defined by a finite word list]
\label{supp:lem:finite-list-containment}
For every fixed finite list $\mathcal S$ of words in the measurement projectors in Definition~\ref{supp:def:finite-word-relaxation}, there exists $d\geq2$ such that
\begin{equation}
\mathcal N_d\subseteq\mathcal N_{\mathcal S},
\label{supp:eq:relaxation-containment}
\end{equation}
where $\mathcal N_d$ denotes the correlation set extracted from the standard level-$d$ relaxation by Lemma~\ref{supp:lem:valid-behavior}.
\end{lemma}

\begin{proof}
Use Eq.~\eqref{supp:eq:abstract-projectors} to expand every projector in every word $S_i$ as a linear combination of $e,a_0,a_1,b_0,b_1$. After distributing products and reducing with the group relations, each $S_i$ is a finite linear combination of reduced words. A word with $D$ projector factors expands into observable words of length at most $D$, and reduction can only shorten them. Thus a finite $d\geq2$ contains every word in all the expansions; it is sufficient to take $d$ at least the largest projector-word length in $\mathcal S$. Enlarge $d$ if necessary to include any additional finitely many moments used to state the relaxation. Any unconstrained moments outside this space may be assigned by a Hermitian linear extension on the reduced-word basis, choosing conjugate values on inverse word pairs. The definition requires no further positive moment matrices involving those moments.

Let $\bm w_d$ be the vector of reduced words of length at most $d$ and write
\begin{equation}
\bm S=T\bm w_d
\label{supp:eq:S-Tw}
\end{equation}
for the real matrix $T$ of expansion coefficients. For any functional feasible at level $d$,
\begin{equation}
\Gamma_{\mathcal S}
=\bigl(\ell(S_i^*S_j)\bigr)_{i,j}
=T\Gamma_d(\ell)T^*\succeq0,
\label{supp:eq:congruence}
\end{equation}
because $v^*\Gamma_{\mathcal S}v=(T^*v)^*\Gamma_d(\ell)(T^*v)\geq0$ for every column vector $v$. The group relations imply all the identities in Eqs.~\eqref{supp:eq:alice-projector-identities}--\eqref{supp:eq:cross-and-adjoint}, and Lemma~\ref{supp:lem:valid-behavior} supplies a valid correlation satisfying the linking relation of Eq.~\eqref{supp:eq:behavior-link}. Hence every correlation in $\mathcal N_d$ belongs to $\mathcal N_{\mathcal S}$.
\end{proof}

\begin{theorem}
\label{supp:thm:set-separation}
For every fixed finite list $\mathcal S$ of words in the measurement projectors defining $\mathcal N_{\mathcal S}$ as in Definition~\ref{supp:def:finite-word-relaxation},
\begin{equation}
\Q\subsetneq\mathcal N_{\mathcal S}.
\label{supp:eq:strict-set-inclusion}
\end{equation}
\end{theorem}

\begin{proof}
Quantum correlations satisfy every finite moment constraint, so $\Q\subseteq\mathcal N_{\mathcal S}$. For the strictness, choose $d\geq2$ as in Lemma~\ref{supp:lem:finite-list-containment}. By Theorem~\ref{supp:thm:fixed-level}, for $\alpha$ sufficiently close to one,
\begin{equation}
\omega_{\rm NPA}^{(d)}(\alpha)>\omega_{\rm Q}(\alpha).
\label{supp:eq:strict-value-gap}
\end{equation}
The primal maximum is attained by Proposition~\ref{supp:prop:duality}; let $\ell$ be an optimizer, and let $p\in\mathcal N_d$ be the correlation extracted from it by Lemma~\ref{supp:lem:valid-behavior}. From Eq.~\eqref{supp:eq:behavior-reconstruction} and the definitions of the abstract projectors,
\begin{align}
\langle A_x\rangle_p&=\ell(a_x),
&\langle B_y\rangle_p&=\ell(b_y),
&\langle A_xB_y\rangle_p&=\ell(a_xb_y),
\label{supp:eq:moment-behavior-identity}
\end{align}
so the Bell value of the extracted correlation is
\begin{equation}
\beta_\alpha(p)=\ell(h_\alpha)=\omega_{\rm NPA}^{(d)}(\alpha)>\omega_{\rm Q}(\alpha).
\label{supp:eq:Bell-value-identity}
\end{equation}
Thus $p\notin\Q$, while Lemma~\ref{supp:lem:finite-list-containment} gives $p\in\mathcal N_{\mathcal S}$, proving strict inclusion.
\end{proof}

The divergent relative error also allows a vanishing weight of the NPA optimizer to raise a deterministic correlation's Bell value above the quantum bound.

\begin{corollary}[Postquantum correlations arbitrarily close to a local correlation]
\label{supp:cor:local-accumulation}
Fix a finite list $\mathcal S$ of words in the measurement projectors as in Definition~\ref{supp:def:finite-word-relaxation}, and let
\begin{equation}
p_{\rm loc}(a,b|x,y)=\delta_{a,1}\delta_{b,1}.
\label{supp:eq:local-deterministic-behavior}
\end{equation}
Every neighborhood of $p_{\rm loc}$ contains a correlation in $\mathcal N_{\mathcal S}\setminus\Q$.
\end{corollary}

\begin{proof}
Choose $d\geq2$ as in Lemma~\ref{supp:lem:finite-list-containment}. For each $\alpha$, let $p_{d,\alpha}\in\mathcal N_d$ attain the level-$d$ value $\omega_{\rm NPA}^{(d)}(\alpha)$, which exists by Proposition~\ref{supp:prop:duality} and Lemma~\ref{supp:lem:valid-behavior}. Separate its excess above the local value into the quantum violation and the NPA error:
\begin{equation}
\Delta_\alpha=\omega_{\rm Q}(\alpha)-\omega_{\rm L}(\alpha),
\qquad
e_{d,\alpha}=\omega_{\rm NPA}^{(d)}(\alpha)-\omega_{\rm Q}(\alpha).
\label{supp:eq:advantage-and-excess}
\end{equation}
For $\alpha$ sufficiently close to one, both are positive and $e_{d,\alpha}/\Delta_\alpha\to\infty$. Mixing weight $\lambda$ increases the local Bell value by $\lambda(\Delta_\alpha+e_{d,\alpha})$. To exceed the quantum bound, this increase must be larger than $\Delta_\alpha$. Choosing $2\Delta_\alpha$ gives
\begin{equation}
\lambda_\alpha=\frac{2\Delta_\alpha}{\Delta_\alpha+e_{d,\alpha}}
\label{supp:eq:mixing-weight}
\end{equation}
Since $e_{d,\alpha}/\Delta_\alpha\to\infty$, this weight lies in $(0,1)$ for $\alpha$ sufficiently close to one and tends to zero. The corresponding correlation is
\begin{equation}
q_\alpha=(1-\lambda_\alpha)p_{\rm loc}+\lambda_\alpha p_{d,\alpha}.
\label{supp:eq:mixed-postquantum-behavior}
\end{equation}
Both $p_{\rm loc}$ and $p_{d,\alpha}$ belong to the convex set $\mathcal N_{\mathcal S}$, the latter by Lemma~\ref{supp:lem:finite-list-containment}. Hence so does $q_\alpha$. Its Bell value is
\begin{align}
\beta_\alpha(q_\alpha)
&=\omega_{\rm L}(\alpha)
+\lambda_\alpha\bigl[\omega_{\rm NPA}^{(d)}(\alpha)-\omega_{\rm L}(\alpha)\bigr]\nonumber\\
&=\omega_{\rm L}(\alpha)+2\Delta_\alpha
=\omega_{\rm Q}(\alpha)+\Delta_\alpha
>\omega_{\rm Q}(\alpha).
\label{supp:eq:mixed-value}
\end{align}
Therefore $q_\alpha\notin\Q$. Finally, all coordinates of $p_{d,\alpha}$ and $p_{\rm loc}$ lie in $[0,1]$, so every coordinate of $q_\alpha-p_{\rm loc}$ has absolute value at most $\lambda_\alpha$. Since $\lambda_\alpha\to0$, we obtain $q_\alpha\to p_{\rm loc}$ in the finite-dimensional correlation space, proving the claim.
\end{proof}

\section{Consequences}
\label{supp:sec:consequences}

The strict inclusion in Theorem~\ref{supp:thm:set-separation} concerns the full correlation, including its four marginal expectations. Keeping only the four joint correlators instead gives a set with an exact finite semidefinite description~\cite{Landau1988,LeEtAl2023}. For the complete set, conjectured extremality criteria and analytical families were studied in Ref.~\cite{MikosKaniewski2023}, its nonlocal extreme points were characterized in Ref.~\cite{BarizienBancal2025}, and their relation to levels $1+AB$ and $2$ was examined in Ref.~\cite{Ishizaka2025}.

Corollary~\ref{supp:cor:local-accumulation} places nonquantum feasible correlations arbitrarily close to a local deterministic point. In particular, the almost quantum set, obtained from the level-$1+AB$ relaxation~\cite{NGHA2015}, strictly contains the complete quantum set already in this scenario. This set satisfies a broad collection of information-theoretic and consistency principles~\cite{NGHA2015}.

The optimal quantum strategies also approach a local limit: their state approaches a product state as both measurement pairs become compatible~\cite{GigenaEtAl2025}.

In the one-sided tilted family the state also approaches a product state, but one party retains a maximally incompatible measurement pair and level $1+AB$ stays exact~\cite{BampsPironio2015,GigenaEtAl2025}.

The certificate used in the proof is dual to moment positivity and the measurement relations. Adding state- or operator-optimality constraints~\cite{AraujoEtAl2026} changes the available certificates, so their effect on this family's finite-level gap requires a separate argument.

\section{Scope}
\label{supp:sec:scope}

For each fixed finite $L$, Theorem~\ref{supp:thm:fixed-level} gives $\delta_L>0$ such that $\omega_{\rm NPA}^{(L)}(\alpha)>\omega_{\rm Q}(\alpha)$ whenever $1-\delta_L<\alpha<1$. The interval may depend on $L$: fixing the word list fixes both the number of certificate coordinates and the constants in their Taylor remainders. The result therefore permits each fixed $\alpha<1$ to have an exact certificate at some higher finite level. Whether one fixed Bell functional in $(2,2,2)$ is missed at every finite level requires a different argument. In a larger scenario, a fixed nonlocal game whose commuting-operator value is not attained at any finite NPA level is known~\cite{FanizzaEtAl2025}.

Finite NPA nonexactness does not imply undecidability of membership in $\Q$. Equation~\eqref{supp:eq:finite-convex-qubit-model} gives at most nine two-qubit components. Their states, projectors, mixture weights, and Born probabilities satisfy finitely many real polynomial constraints; positivity is expressed by nonnegative values of all principal minors. When the specified probabilities are rational numbers or exactly specified real algebraic numbers (real roots of integer polynomials), membership is therefore decidable by real quantifier elimination~\cite{Tarski1951}, as also noted in Ref.~\cite{Masanes2005}. This need not be an efficient procedure.

The same reasoning applies to distance from the set. For rational or algebraic input data $p$ and $\varepsilon\geq0$, the condition that the Euclidean distance from $p$ to $\Q$ is at most $\varepsilon$ is equivalent, by compactness, to the existence of $q\in\Q$ with
\begin{equation}
\sum_{a,b,x,y}\bigl[p(a,b|x,y)-q(a,b|x,y)\bigr]^2\leq\varepsilon^2.
\label{supp:eq:distance-decision}
\end{equation}
This is another finite polynomial feasibility problem. The result is therefore distinct from undecidability for unrestricted nonlocal games~\cite{MIPstarRE}.

The polynomial constraints above include products of state, measurement, and mixture variables, so this finite decision procedure does not itself give a semidefinite representation. More generally, a finite polynomial description does not guarantee a finite semidefinite representation~\cite{Scheiderer2018}. The theorem excludes the standard finite NPA descriptions of $\Q$, leaving the further question posed in the Letter: does the complete quantum set in the simplest Bell scenario admit an exact finite semidefinite representation?

Two preprints by Pakhunov, submitted to arXiv on 15 July 2026 after our submission to Physical Review Letters on 14 July 2026, claim to prove finite-level nonexactness for the same doubly tilted family and to identify the same Motzkin limit~\cite{Pakhunov2026a,Pakhunov2026b}. The finite-level argument is described there as computer-assisted, with finite exact-arithmetic verifications forming part of the proof. We have not independently verified those computations or the claimed bound. No statement or proof in this Supplemental Material relies on either preprint.